\documentclass[11pt,a4paper]{amsart}
\usepackage[utf8]{inputenc}
\usepackage{amsmath}
\usepackage{amsfonts}
\usepackage{amssymb}
\usepackage{graphicx}
\usepackage{booktabs}
\usepackage{tikz}
\usepackage[ruled,vlined]{algorithm2e}
\usepackage{floatrow}
\usepackage[round,authoryear]{natbib}
\usepackage{bibunits}
\usepackage{nicefrac}
\newfloatcommand{capbtabbox}{table}[][\FBwidth]

\usepackage{blindtext}

\usepackage[left=2.2cm,right=2.2cm,top=2cm,bottom=2cm]{geometry}

\newcommand{\R}{\mathbb{R}}

\newcommand{\K}{\mathbb{K}}
\newcommand{\N}{\mathbb{N}}
\newcommand{\NN}{\mathcal{N}\mathcal{N}}

\numberwithin{equation}{section}  
\newtheorem{defn}{Definition}[section]

\newtheorem{rem}[defn]{Remark}
\newtheorem{thm}[defn]{Theorem}
\newtheorem{prop}[defn]{Proposition}
\newtheorem{cor}[defn]{Corollary}
\newtheorem{lem}[defn]{Lemma}

\newtheorem{asu}[defn]{Assumption}

\usepackage{hyperref}
\usepackage{xcolor}
\hypersetup{
    colorlinks,
    linkcolor={red!50!black},
    citecolor={blue!50!black},
    urlcolor={blue!80!black}
}

\title[Neural Networks can detect model-free static arbitrage strategies]{Neural Networks can detect \\model-free static arbitrage strategies}

\author[A. Neufeld, J. Sester]{ Ariel Neufeld$^{1}$, Julian Sester$^{2}$}

\begin{document}

\maketitle

\begin{center}
\normalsize{\today} \\ \vspace{0.5cm}
\small\textit{$^{1}$NTU Singapore, Division of Mathematical Sciences,\\ 21 Nanyang Link, Singapore 637371.\\
$^{2}$National University of Singapore, Department of Mathematics,\\ 21 Lower Kent Ridge Road, 119077.}                                                                                                                              
\end{center}

\begin{abstract}~
In this paper we demonstrate both theoretically as well as numerically
that neural networks can detect model-free static arbitrage opportunities whenever the market admits some. Due to the use of neural networks, our method can be applied to
financial markets with a high number of traded securities and ensures almost immediate execution of the corresponding trading strategies. To demonstrate its tractability, effectiveness, and robustness we provide examples using real financial data. 
From a technical point of view, we prove that a \textit{single} neural network can approximately solve a \textit{class} of convex semi-infinite programs, which is the key result in order to derive our theoretical results that neural networks can detect model-free static arbitrage strategies whenever the financial market admits such opportunities.
%
%
%
\end{abstract}

\section{Introduction}

Detecting arbitrage opportunities in financial markets and efficiently implementing them numerically is an intricate and demanding task, both in theory and practice. In recent academic papers, researchers have extensively tackled this problem for various types of assets, highlighting its significance and complexity.

The authors from \cite{cui2020detecting} and \cite{cui2020arbitrage} focus their studies on the foreign exchange market, establish conditions that eliminate triangular opportunities and propose computational approaches to detect arbitrage opportunities.

\cite{soon2011currency} propose a binary integer programming model for the detection of arbitrage in currency exchange markets, while   \cite{papapantoleon2021detection} focus on arbitrage in multi-asset markets under the assumptions that the risk-neutral marginal distributions are known. 
Also assuming knowledge of risk-neutral marginals in multi-asset markets, \cite{tavin2015detection} provides a copula-based approach to characterize the absence of arbitrage.
\cite{cohen2020detecting} study arbitrage opportunities in markets where vanilla options are traded and propose an efficient procedure to change the option prices minimally (w.r.t.\,the $\l^1$ distance) such that the market becomes arbitrage-free. \cite{neufeld2022model} develop cutting-plane based algorithms to calculate model free upper and lower price bounds 
 whose sub-optimality can be chosen to be arbitrarily small, and use them to detect model-free arbitrage strategies. Furthermore, by observing call option prices  \cite{biagini2022detecting} train neural networks to detect financial asset bubbles.

In this paper we study the detection of model-free static arbitrage in potentially high-dimensional financial markets,
 i.e., in markets where a large number of securities are traded.
 A trading strategy is called \textit{static} if the strategy consists of buying or selling financial derivatives as well as the corresponding underlying securities in the market  only at initial time (with corresponding bid and ask prices) and then holding the positions till maturity without any readjustment. Therefore, one says that a market admits static arbitrage if there exists a static trading strategy which provides a guaranteed risk-free profit at maturity.
We aim to detect static arbitrage opportunities in a \textit{model-free} way, i.e.\ purely based on observable market data without imposing any (probabilistic) model assumptions on the underlying financial market. 
We also refer to \cite{acciaio2016model, Burzoni,burzoni2019pointwise, burzoni2021viability, cheridito2017duality, davis2014arbitrage, fahim2016model, hobson2005static,hobson2005static2, neufeld2022deep, riedel2015financial, wang2021necessary} for more details on model-free arbitrage and its characterization. 

The goal of this paper is to demonstrate both theoretically as well as numerically using real-market data  
that neural networks can detect model-free static arbitrage whenever the market admits some. The motivation of using neural networks is their known ability to efficiently deal with high-dimensional problems in various fields. 
There are several algorithms that can detect (static) arbitrage strategies in a financial market under fixed market conditions, for example in a market with fixed options with corresponding strikes as well as fixed corresponding bid and ask prices. However
 directly applying these algorithms in real financial market scenarios to exploit arbitrage is challenging due to well-known issue that market conditions changes extremely fast and high-frequency trading often cause these opportunities to vanish rapidly.  This associated risk is commonly known as \emph{execution risk}, as discussed, e.g., in \cite{kozhan2012execution}. The speed of investment execution therefore becomes crucial in capitalizing on arbitrage opportunities. 

By training neural networks according to our algorithm purely based on observed market data, we obtain \emph{detectors} that allow, given any market conditions, to detect not only the existence of static arbitrage but also to determine a proper applicable arbitrage strategy. Our algorithm therefore provides to financial agents an instruction how to trade and to exploit the arbitrage strategy while the opportunity persists. In contrast to other numerical methods which need to be executed entirely each time the market is scanned for arbitrage or the market conditions are changing, our proposed method only needs one neural network to be trained offline. After training, the neural network is then able to detect arbitrage and can be executed extremely fast, allowing to invest in the resultant strategies in every new market situation that one faces. We refer to Section~\ref{sec_exa} for detailed description of our algorithm as well as our numerical results evaluated on real market data.

We justify the use of neural networks by proving that neural networks can detect model-free static arbitrage strategies whenever the market admits some. We refer to Theorem~\ref{cor_ftpa} and Theorem~\ref{thm:epsilon-arbitrage} for our main theoretical results regarding arbitrage detection. The main idea is to relate arbitrage with the superhedging of the zero-payoff function. We prove in Proposition~\ref{prop_arbitrage} that there exists a \textit{single} neural network that provides a corresponding $\varepsilon$-optimal superhedging strategy \textit{for any given market conditions}. In fact, we show for a certain class of convex semi-infinite programs (CSIP), which includes the superhedging problem of the zero-payoff function as special case, that a \textit{single} neural network can provide \textit{for each of  the (CSIP)} within this class a corresponding feasible and $\varepsilon$-optimal solution, see Theorem~\ref{thm_main}.

The remainder of this paper is as follows. In Section~\ref{sec_arbitrage}, we introduce the setting of the financial market as well as the corresponding (static) trading strategies, and provide our main theoretical results ensuring that model-free static arbitrage can be detected by neural networks if existent. Section~\ref{sec_exa} focuses on the presentation and numerical implementation of our neural networks based algorithm to detect static arbitrage, featuring experiments conducted on real financial data to showcase the feasibility and robustness of our method. In Section~\ref{sec_convex}, we introduce a class of convex semi-infinite programs and provide our main technical result that a single neural network can approximately solve this class of (CSIPs). Finally, all proofs are presented in Section~\ref{sec_proofs}.

\section{Detection of static Arbitrage Strategies}\label{sec_arbitrage}
In this section, we study a financial market  in which a financial agent can trade statically in various types of options and which may admit the opportunity of static arbitrage profits.
In such a setting, the natural difficulty for a trader is first to decide whether such arbitrage exists and second to identify potential strategies that exploit arbitrage profits. Our goal is to show that for each financial market in which an agent can trade statically in options, the corresponding market admits static arbitrage if and only if there exists a neural network that detects the existence of model-free static arbitrage by outputting a corresponding arbitrage strategy. 

\subsection{Setting}\label{sec:setting}
In this paper, we consider a market in which a financial agent can trade statically in options. To introduce the market under consideration,
let $S=(S_1,\dots,S_d)$ denote the underlying $d\in \N$ stocks at some future time $t=1$.
We only consider values $S \in \mathcal{S} \subseteq [0,\infty)^d$ for some predefined set $\mathcal{S}$, which can be interpreted as prediction set\footnote{We refer to, e.g., \cite{bartl2020pathwise,hou2018robust, mykland2003financial,neufeldsester2021model} for further literature on prediction sets in financial markets.} where the financial agent  may allow to exclude values which she considers to
be impossible to model future stock prices $S=(S_1,\dots,S_d)$ at time $1$.

Let $N_{\Psi} \in \N$ denote the number of different \emph{types}\footnote{We say two options are of the same \emph{type} if the payoffs only differ  with respect to the specification of a strike. Also note that trading in the underlying securities itself can be considered as an option, e.g., a call option with strike $0$.} of traded options $\Psi_i: \mathcal{S} \times [0,\overline{K}] \rightarrow [0, \infty)$, $i=1,\dots,N_{\Psi}$, written on $S$. For each option type $i\in \{1,\dots,N_{\Psi}\}$ let
 $n_i\in \N$ 
denote the corresponding amount of different strikes under consideration $(K_{i,j})_{j=1,\dots,n_i} \subseteq [0,\overline{K}]$, where the strikes are contained in $[0,\overline{K}]$ for some  $\overline{K}<\infty$, and denote by $N:= \sum_{i=1}^{N_\Psi}n_i$ the total number of traded options.
Moreover, we denote by $\pi=(\pi_{i,j})_{i=1,\dots,N_{\Psi}, \atop j=1,\dots,n_i} =(\pi_{i,j}^+,\pi_{i,j}^-)_{i=1,\dots,N_{\Psi}, \atop j=1,\dots,n_i} \in [0,\overline{\pi}]^{2N}$ the bid and ask prices of the traded options respectively,  where we assume that all the bid and ask prices are bounded by some $\overline{\pi}>0$.
%

The financial agent then can trade in the market  by buying and selling the options described above. More precisely, we first fix the minimal initial cash position of a trading strategy to be given by $\underline{a}\in \R$, and we assume that the maximal amount of shares of options one can buy or sell  is capped by some constant $0<\overline{H}< \infty$. This allows to consider the payoff of a static trading strategy by the function
\begin{equation}\label{eq_defn_Is_arbitrage_setting}
\begin{aligned}
\mathcal{I}_S: [0,\overline{K}]^N \times [\underline{a},\infty) \times [0, \overline{H}]^{2N} &\rightarrow \R \\
(K,a,h) &\mapsto a+\sum_{i=1}^{N_{\Psi}} \sum_{j=1}^{n_i} \left(h_{i,j}^+-h_{i,j}^-\right) \cdot \Psi_i(S,K_{i,j}),
\end{aligned}
\end{equation}
where we use the notation $h= \left(h_{i,j}^+,h_{i,j}^-\right)_{i =1,\dots,N_{\Psi}, \atop j = 1,\dots,n_i} \in [0, \overline{H}]^{2N}$ to denote long and short positions in the traded options, respectively, as well as $K=(K_{i,j})_{i=1,\dots,N_{\Psi}, \atop j=1,\dots,n_i}$ to denote all strikes.
The corresponding pricing functional is then defined by
\begin{equation}\label{eq_defn_f_arbitrage_setting}
	\begin{aligned}
		f:[0, \overline{\pi}]^{2N} \times  [\underline{a},\infty)\times [0,\overline{H}]^{2N} &\rightarrow \R\\
		\left(\pi,a,h \right) &\mapsto a+ \sum_{i=1}^{N_{\Psi}} \sum_{j=1}^{n_i} \left(h_{i,j}^+\pi_{i,j}^+-h_{i,j}^-\pi_{i,j}^- \right)
	\end{aligned}
\end{equation}
 determining the price of a corresponding trading strategy with respect to the corresponding bid and ask prices of the options.

Moreover, we define the set-valued map which maps a set of strikes $K=(K_{i,j})_{i=1,\dots,N_{\Psi}, \atop j=1,\dots,n_i}$ to the corresponding strategies leading to a greater payoff than $0$ for each possible value $S\in \mathcal{S}$ by
\begin{equation}
\label{eq_defn_Gamma_arbitrage_setting}
\Gamma: [0,\overline{K}]^N \ni K \twoheadrightarrow \Gamma(K):= \left\{ (a,h) \in  [\underline{a},\infty) \times [0,\overline{H}]^{2N} ~\middle|~ \mathcal{I}_S(K,a,h) \geq 0 \text{ for all } S \in \mathcal{S}\right\}.
\end{equation}

In this paper, we consider the following type of model-free\footnote{It is called \textit{model-free} since no probabilistic assumptions on the financial market has been imposed} static arbitrage. We refer to \cite{burzoni2019pointwise} for several notions of model-free arbitrage.
\begin{defn}[Model-free static arbitrage] \label{def_arbitrage}
Let $(K,\pi) \in [0, \overline{K}]^{N}\times [0, \overline{\pi}]^{2N}$. We call a static trading strategy $(a,h) \in  [\underline{a},\infty) \times [0, \overline{H}]^{2N}$ a \emph{model-free static arbitrage strategy} if the following two conditions hold.
\begin{itemize}
\item[(i)] $(a,h)\in \Gamma(K)$,
\item[(ii)] $f(\pi,a,h)<0$.
\end{itemize}
Moreover, for any $\varepsilon>0$ we call a model-free static arbitrage strategy to be of \textit{magnitude $\varepsilon$} if $f(\pi,a,h)\leq -\varepsilon$. 
\end{defn}

Then, the minimal price of a trading strategy that leads to a greater price than $0$ for each possible value $S\in \mathcal{S}$, in dependence of strikes $K=(K_{i,j})_{i=1,\dots,N_{\Psi}, \atop j=1,\dots,n_i}$ and option prices $\pi =(\pi_{i,j}^+,\pi_{i,j}^-)_{i=1,\dots,N_{\Psi}, \atop j=1,\dots,n_i}$, is given by 
\begin{equation}\label{eq_defn_Psi_arbitrage_setting}
\begin{aligned}
V:[0, \overline{K}]^{N}\times [0, \overline{\pi}]^{2N} &\rightarrow \R\\
(K,\pi) &\mapsto \inf_{(a,h) \in \Gamma(K)} f(\pi,a,h).
\end{aligned}
\end{equation}

This means according to Definition~\ref{def_arbitrage} that the market with parameters $(K,\pi) \in [0, \overline{K}]^{N}\times [0, \overline{\pi}]^{2N}$ admits no model-free static arbitrage strategy if and only if $V(K,\pi) \geq 0$.
\subsection*{Neural Networks}
By \textit{neural networks} with input dimension $d_{\operatorname{in}} \in \N$, output dimension $d_{\operatorname{out}} \in \N$, and number of layers $l \in \N$ we refer to functions of the form
\begin{equation}\label{eq_nn_function}
	\begin{aligned}
		\R^{d_{\operatorname{in}}} &\rightarrow \R^{d_{\operatorname{out}}}\\
		{x} &\mapsto {A_l} \circ {\varphi}_l \circ {A_{l-1}} \circ \cdots \circ {\varphi}_1 \circ {A_0}({x}),
	\end{aligned}
\end{equation}
where $({A_i})_{i=0,\dots,l}$ are affine\footnote{This means for all $i=0,\dots,l$, the function ${A_i}$ is assumed to have an affine structure of the form
	$
	{A_i}({x})={M_i} {x} + {b_i}
	$
	for some matrix ${M_i} \in \R^{ h_{i+1} \times h_{i}}$ and some vector ${b_i}\in \R^{h_{i+1}}$, where $h_0:=d_{\operatorname{in}}$ and $h_{l+1}:=d_{\operatorname{out}}$. } functions of the form 
\begin{equation}\label{eq_A_i_def}
	{A_0}: \R^{d_{{\operatorname{in}}}} \rightarrow \R^{h_1},\qquad {A_i}:\R^{h_i}\rightarrow \R^{h_{i+1}}\text{ for } i =1,\dots,l-1, \text{(if } l>1), \text{ and}\qquad {A_l} : \R^{h_l} \rightarrow \R^{d_{\operatorname{out}}},
\end{equation}
and where the function $\varphi_i$ is applied componentwise, i.e., for  $i=1,\dots,l$ we have ${\varphi}_i(x_1,\dots,x_{h_i})=\left(\varphi(x_1),\dots,\varphi(x_{h_i})\right)$.  The function $\varphi:\R \rightarrow \R$  is called \emph{activation function} and assumed to be continuous and non-polynomial.
We say a neural network is \emph{deep} if $l\geq 2$.
Here ${h}=(h_1,\dots,h_{l}) \in \N^{l}$ denotes the dimensions (the number of neurons) of the hidden layers, also called \emph{hidden dimension}.

Then, we denote by $\mathfrak{N}_{d_{\operatorname{in}},{d_{\operatorname{out}}}}^{l,{h}}$  the set of all neural networks with input dimension ${d_{\operatorname{in}}}$, output dimension ${d_{\operatorname{out}}}$, $l$ hidden layers, and hidden dimension ${h}$, whereas
the set of all neural networks from $\R^{d_{\operatorname{in}}}$ to $\R^{d_{\operatorname{out}}}$ (i.e.\ without specifying the number of hidden layers and hidden dimension) is denoted by
\[
\mathfrak{N}_{d_{\operatorname{in}},{d_{\operatorname{out}}}}:=\bigcup_{l \in \N}\bigcup_{{h} \in \N^l}\mathfrak{N}_{d_{\operatorname{in}},{d_{\operatorname{out}}}}^{l,{h}}.
\]
%

It is well-known that the set of neural networks possess the so-called \textit{universal approximation property}, see, e.g., \cite{pinkus1999approximation}.
\begin{prop}[Universal approximation theorem]\label{lem_universal}
For any compact set $\K \subset \R^{d_{\operatorname{in}}} $ the set $\mathfrak{N}_{d_{\operatorname{in}},{d_{\operatorname{out}}}}|_{\K}$ is dense in ${C}(\K,\R^{d_{\operatorname{out}}})$ with respect to the topology of uniform convergence on $C(\K,\R^{d_{\operatorname{out}}})$.
\end{prop}
\subsection{Main results}
To formulate our main result we first impose the following mild assumptions.
\begin{asu}\label{asu_arbitrage}~
\begin{itemize}
\item[(i)] There exists some $L_{\Psi}>0$ such that for all $S \in \mathcal{S}$ and for all $i=1,\dots, N_{\Psi}$ the map $[0,\overline{K}] \ni K_{i,j} \mapsto \Psi_i(S,K_{i,j})$ is $L_{\Psi}$-Lipschitz.
\item[(ii)] There exists some by $C_\Psi>0$ such that  the map $\Psi_i$ is bounded by $C_{\Psi}$ on $\mathcal{S} \times [0,\overline{K}]$ for all $i=1,\dots, N_{\Psi}$.
\end{itemize}
\end{asu}

\begin{rem}
	First, note that we do not impose any topological or geometric conditions on the prediction set $\mathcal{S} \subset [0,\infty)^d$.
However, a sufficient criterion for Assumption~\ref{asu_arbitrage}~(ii) to hold would be that, e.g.,  $\mathcal{S} \subset [0,\infty)^d$ is bounded  and that $[0,\infty)^d \times [0,\overline{K}] \ni (S,K) \mapsto \Psi_i(S,K)$ is continuous for each $i=1,\dots,N_{\Psi}$. Moreover, note that Assumption~\ref{asu_arbitrage}~(i) is satisfied for example for any payoff function which is continuous and piece-wise affine (CPWA), which includes most relevant payoff functions in finance. We refer to \cite{neufeld2022model,li2023quantum} for a detailed list of examples of (CPWA) payoff functions.
\end{rem}
In our first result, we conclude that the financial market described in Section~\ref{sec:setting} admits model-free static arbitrage if and only if there exists a neural network that detects the existence of model-free static arbitrage by outputting a corresponding arbitrage strategy.
\begin{thm}[Neural networks can detect static arbitrage]\label{cor_ftpa}
	Let Assumption~\ref{asu_arbitrage} hold true, and let $(K,\pi)\in [0, \overline{K}]^{N} \times [0, \overline{\pi}]^{2N}$.
	Then, there exists model-free static arbitrage if and only if there exists a neural network $ \NN \in  \mathfrak{N}_{3N,1+2N}$ with 
	\begin{itemize}
		\item[(i)] $\NN(K,\pi):=\left(\NN_a(K,\pi),\NN_h(K,\pi)\right) \in \Gamma(K)$,
		\item[(ii)] $f\left(\pi,\NN_a(K,\pi),\NN_h(K,\pi)\right)<0$.
	\end{itemize}
\end{thm}

In our second result, we show that for any given $\varepsilon>0$ and $0< \delta< \varepsilon$ there exists a \textit{single} neural network such that for any given strikes $K=(K_{i,j})_{i=1,\dots,N_{\Psi}, \atop j=1,\dots,n_i}$ and option prices $\pi =(\pi_{i,j}^+,\pi_{i,j}^-)_{i=1,\dots,N_{\Psi}, \atop j=1,\dots,n_i}$  the neural network can detect model-free static arbitrage of magnitude $\delta$ if the financial market with corresponding market conditions $(K,\pi)$ admits static arbitrage of magnitude $\varepsilon$. From a practical point of view, this is crucial, since it allows the financial trader to \textit{only train one single neural network} which can then, once trained, instantaneously detect corresponding static arbitrage opportunities if the current market conditions $(K,\pi)$ admit such opportunities. On the other hand, a trader applying the trained neural network to a financial market which admits no static arbitrage opportunities pays at most $\varepsilon-\delta$ for the trading strategy, i.e., if $\varepsilon \approx \delta$, the risk of paying for trading strategies which are no static arbitrage strategies can be reduced to an arbitrarily small amount.

\begin{thm}[A single neural network can detect static arbitrage of magnitude $\varepsilon$]\label{thm:epsilon-arbitrage}
	Let $\varepsilon>0$ and $0<\delta<\varepsilon$. Then, there exists a neural $ \NN \in  \mathfrak{N}_{3N,1+2N}$ such that for every  $(K,\pi)\in [0, \overline{K}]^{N} \times [0, \overline{\pi}]^{2N}$ the following holds.
\begin{itemize}
\item[(i)]	If the financial market with respect to $(K,\pi)$ admits  model-free static arbitrage of magnitude $\varepsilon$, then the neural network outputs a trading strategy $\left(\NN_a(K,\pi),\NN_h(K,\pi)\right)$ which is a model-free static arbitrage of magnitude $\delta$.
\item[(ii)]	 If the financial market with respect to $(K,\pi)$ admits no model-free static arbitrage, then the neural network outputs a trading strategy $\left(\NN_a(K,\pi),\NN_h(K,\pi)\right) \in \Gamma(K)$ which has a price of at most $\varepsilon-\delta$.
\end{itemize}
\end{thm}

The main idea to derive Theorem~\ref{cor_ftpa} and Theorem~\ref{thm:epsilon-arbitrage} relies on the relation between arbitrage and superhedging of the $0$-payoff function. The following result establishes that for any prescribed $\varepsilon>0$ there exists a \textit{single} neural network such that for any given strikes $K=(K_{i,j})_{i=1,\dots,N_{\Psi}, \atop j=1,\dots,n_i}$ and option prices $\pi =(\pi_{i,j}^+,\pi_{i,j}^-)_{i=1,\dots,N_{\Psi}, \atop j=1,\dots,n_i}$ defining the  market, the neural network produces a static trading strategy which superhedges the $0$-payoff for all possible values $S \in \mathcal{S}$ whose price is $\varepsilon$-optimal. 
\begin{prop}[Approximating $V$ with neural networks]\label{prop_arbitrage}
Let Assumption~\ref{asu_arbitrage} hold true. Then for all $\varepsilon>0$ there exists a neural network  $\NN \in  \mathfrak{N}_{3N,1+2N} $ such that 
\begin{itemize}
\item[(i)] $\NN(K,\pi):=\left(\NN_a(K,\pi),\NN_h(K,\pi)\right) \in \Gamma(K)$ for all $(K,\pi) \in [0, \overline{K}]^{N} \times [0, \overline{\pi}]^{2N}$,
\item[(ii)] $f\left(\pi,\NN_a(K,\pi),\NN_h(K,\pi)\right)-V(K,\pi) \leq \varepsilon$ for all $(K,\pi) \in [0, \overline{K}]^{N} \times [0, \overline{\pi}]^{2N}$.
\end{itemize}
\end{prop}
In fact, we will use Proposition~\ref{prop_arbitrage} to prove our main results Theorem~\ref{cor_ftpa} and Theorem~\ref{thm:epsilon-arbitrage} on detecting static arbitrage strategies. To prove Proposition~\ref{prop_arbitrage}, we interpret \eqref{eq_defn_Psi_arbitrage_setting} as a class of linear semi-infinite optimization problem (LSIP), where each $(K,\pi) \in [0, \overline{K}]^{N} \times [0, \overline{\pi}]^{2N}$ determines a single (LSIP). In Section~3, we introduce a (much more) general class of convex semi-infinite optimization problem (CSIP) which covers \eqref{eq_defn_Psi_arbitrage_setting} as special case. Then we show that a single neural network can approximately solve all (CSIP) of this class simultaneously. We refer to Theorem~\ref{thm_main} for the precise statement.   

%
%
%
\noindent
The proofs of all our main results  are provided in Section~\ref{sec_proofs}.

\section{The Numerics of Static Arbitrage Detection in Financial Markets}\label{sec_exa}

 The results from Section~\ref{sec_arbitrage} prove, with non-constructive arguments, the existence of neural networks that can detect model-free arbitrage strategies. These results therefore immediately raise the question how to construct neural networks that are capable to learn these strategies. 
To this end, we present with Algorithm~\ref{algo_3} an approach that combines a supervised learning approach in the spirit of \cite{neufeld2022deep} with an unsupervised learning approach as presented for example in \cite{auslender2009penalty}, \cite{eckstein2021computation}, and \cite{neufeld2022detecting}.

Algorithm~\ref{algo_3} uses the fact that in many situations there exists an applicable algorithm to compute model-free price bounds and corresponding trading strategies that approximate these bounds arbitrarily well. We exploit this fact by training a neural network offline to approximate the outcomes of such algorithms. To compute the strategies that approximately attain these bounds, we suggest employing the algorithm presented in \cite{neufeld2022model}.
The motivation of our methodology is the following. 
While the offline training of the neural network might take some time, once trained, the neural network is able to detect \textit{immediately} static arbitrage and the corresponding trading strategies in the market, provided it exists. This is crucial as stock prices and corresponding option prices move quickly in real financial markets and therefore having an algorithm which can adjust fast to new market parameters is desired. \newpage

\begin{algorithm}[h!]\label{algo_3}
\SetAlgoLined
\SetKwInOut{Input}{Input}
\SetKwInOut{Output}{Output}
\Input{Number of iterations $N_{\operatorname{iter}}$, Hyperparameters of the neural network, Number of options $N\in \N$,   Option payoffs $(\Psi_i)_i$, Bounds $\overline{K}$, $\overline{\pi}$, $\overline{H}>0$, $\underline{a}\in \R$; Penalization parameter $\gamma >0$, Batch Size $B$ of samples, Batch Size $S_B$ for outcomes of underlying assets ;}

Initialize a neural network $\NN=(\NN_a,\NN_h): \R^{3N} \rightarrow \R^{1+2N}$  with a bounded output layer  such that $\NN$ attains values\footnotemark \ in $[\underline{a},\infty] \times [0, \overline{H}]^{2N} \subset \R^{1+2N}$;\\
\For {$i=1,\dots,N_{\operatorname{iter}}$} {
\For {$b=1,\dots,B$} {
Sample $X_{i,b}: = (K_{i,b},\pi_{i,b})\in [0,\overline{K}]^N\times [0,\overline{\pi}]^{2N}$;\\
Compute $(a_{i,b},h_{i,b}) \in \Gamma(K_i)$ such that 
$Y_{i,b} := f(\pi_{i,b},a_{i,b},h_{i,b}) \approx V(\pi_{i,b},K_{i,b})$, e.g., according to \cite{neufeld2022model};\\
Set $\widetilde{Y}_{i,b} := \begin{cases}
-1 &\text{ if } Y_{i,b} < 0,\\
0 &\text{ if }  Y_{i,b} \geq 0.
\end{cases}$;\\
\For {$j=1, \dots , S_B$}{
Sample $S_{i,b,j} \in \mathcal{S}$;\\
}
}
}

\For {$i=1,\dots,N_{\operatorname{iter}}$} {

Minimize 
\begin{align*}
\sum_{b=1}^B  \bigg \{ &f \left(\pi_{i,b},\NN_a(K_{i,b},\pi_{i,b}),\NN_h(K_{i,b},\pi_{i,b})\right)\\
&+\gamma \cdot   \frac{1}{S_B} \sum_{j=1}^{S_B} \left(\left(-\mathcal{I}_{S_{i,b,j} }\left(K_{i,b},\NN_a(K_{i,b},\pi_{i,b}),\NN_h(K_{i,b},\pi_{i,b}\right)\right)^+\right)^2\\
&+ \gamma \cdot  \bigg(-(\widetilde{Y}_{i,b}+0.5)\cdot f \left(\pi_{i,b},\NN_a(K_{i,b},\pi_{i,b}),\NN_h(K_{i,b},\pi_{i,b})\right)\bigg)^+\bigg\}
\end{align*}
w.r.t.\,the parameters of  $\NN_a$ and $\NN_h$;
}
\Output{A trained neural network $\NN$;}
\caption{Training of an arbitrage-detecting neural network.}
\end{algorithm}
 \footnotetext{This can be realized, e.g., by $\operatorname{tanh}$ and $\operatorname{sigmoid}$ activation functions multiplied with the corresponding bounds.}
Algorithm~\ref{algo_3} is designed to minimize the price function $f \left(\pi_{i,b},\NN_a(K_{i,b},\pi_{i,b}),\NN_h(K_{i,b},\pi_{i,b})\right)$ by incorporating two specific penalization terms. These terms are carefully crafted to facilitate the learning of the key characteristics associated with model-free arbitrage strategies.

 The first penalization term \footnote{We denote by $x^+ = \max\{x,0\}$ the positive part of a real number $x \in \R$.} $\gamma \cdot   \frac{1}{S_B} \sum_{j=1}^{S_B} \left(\left(-\mathcal{I}_{S_{i,b,j} }\left(K_{i,b},\NN_a(K_{i,b},\pi_{i,b}),\NN_h(K_{i,b},\pi_{i,b}\right)\right)^+\right)^2$ incentivizes the feasibility (see \eqref{eq_defn_Gamma_arbitrage_setting}) of learned strategies by penalizing negative payoffs in proportion to the degree of violation of the positivity constraint. This encourages the strategies to have positive payoffs.

 The second penalization term 
 \begin{equation}\label{eq_penal_2}
\gamma \cdot  \bigg(-(\widetilde{Y}_{i,b}+0.5)\cdot f \left(\pi_{i,b},\NN_a(K_{i,b},\pi_{i,b}),\NN_h(K_{i,b},\pi_{i,b})\right)\bigg)^+
 \end{equation}
 vanishes if and only if the price $ f \left(\pi_{i,b},\NN_a(K_{i,b},\pi_{i,b}),\NN_h(K_{i,b},\pi_{i,b})\right)$ of the strategy expressed by the neural network and the pre-computed price $\widetilde{Y}_{i,b}$ are either both non-negative, or both negative.
 Since the  pre-computed price $\widetilde{Y}_{i,b}$ is negative if and only if the market (under the current market parameters $(K_{i,b},\pi_{i,b})$) admits some static arbitrage,  
  the second penalization term vanishes if and only if the trading strategy expressed by the neural network correctly identifies if the markets admits static arbitrage, or not.

It is worth mentioning that the design of the penalization terms does not guarantee the feasibility of strategies in the sense of \eqref{eq_defn_Gamma_arbitrage_setting} or the correct sign of prices. However, due to the penalty imposed on constraint violations, as demonstrated in Example~\ref{sec_exa_1}, in practice, violations happen frequently but are typically only marginal in magnitude.

\begin{rem}
Algorithm~\ref{algo_3} is designed to fulfill two tasks simultaneously: first, to detect whether there  exists arbitrage in the market (through \eqref{eq_penal_2}) and second how to exploit arbitrage if existent, while the focus through the design of the objective function lies on maximizing the profit if arbitrage exists. For the pure classification task of deciding whether arbitrage exists (without learning the associated arbitrage strategy), there exist better suited classification algorithms such as Logistic Regression (\cite{kleinbaum2002logistic}), random forests (\cite{ho1995random}) or XGBoost (\cite{chen2016xgboost}). We leave the exploration of optimizing this related but still different task for future research.
\end{rem}

\subsection{Application to real financial data} 

In the following we apply Algorithm~\ref{algo_3} to real financial data in order to detect model-free static arbitrage in the trading of financial derivatives. For convenience of the reader, we provide under \href{https://github.com/juliansester/Deep-Arbitrage}{https://github.com/juliansester/Deep-Arbitrage} the used Python-code.

\subsubsection{Training with data of the S\&P 500 }\label{sec_exa_1}

We consider trading in a financial market that consists of $d=5$ assets and corresponding $10$ vanilla call options (i.e.\,$10$ different strikes) written on each of the assets. This means we consider $N_\Psi = 5$ different types of options with
$n_i = 11$ for $i=1,\dots,5$ referring to the number of call options plus the underlying assets (which can be considered as a call option with strike $0$) so that in total $N = \sum_{i=1}^{N_{\Psi}} n_i = 55$ different securities are considered. 

To create a training set we consider for each of the $500$ constituents of the $S \& P~500$ the $10$ most liquidly traded\footnote{"Most liquidly traded" refers to the strikes with the highest trading volume.} call options with maturity $T = 19$ May $2023$. The data was downloaded on $25$ April $2023$ via \emph{Yahoo Finance}.

We then use this data to create $50~000$ samples by combining the call options of $5$ randomly chosen constituents in each sample. The spot values of the underlying assets are scaled to $1$, therefore the strikes and corresponding prices are included as percentage values w.r.t.\,the spot value of the underlying asset. We assume $\mathcal{S}= [0,2]^5$, i.e, we assume that the underlying assets at maturity only attain values between $0\%$ and $200\%$ of its current spot value. This assumption can be regarded as a restriction imposed on the space of possible outcomes to a prediction set as mentioned in the beginning of Section~\ref{sec:setting}.

 Relying on these samples, we compute, using the LSIP algorithm from \cite{neufeld2022model}, minimal super-replication strategies of the $0$-payoff for each of the $50~000$ samples.

Of these $50~000$ samples, we regard $5000$ samples as a test set on which the neural network is not trained.

To demonstrate the performance of our approach, we apply Algorithm~\ref{algo_3} with $N_{\operatorname{iter}} = 20~000$ iterations, a penalization parameter\footnote{Following the empirical experiments from  \cite{eckstein2021robust} and \cite{eckstein2021computation}, in the implementation, we let $\gamma$ increase with the number of iterations so that in the first iteration $\gamma$ equal $1$, and after $20~000$ iterations $\gamma$ is $10~000$.} $\gamma = 10~000$, and batch sizes $S_B = 32$ and $B=512$ to train a neural network $\NN$ with $1024$ neurons and $5$ hidden layers and with a \emph{ReLU} activation function in each of the hidden layers. The used learning rate for training with the \emph{Adam} optimizer (\cite{kingma2014adam}) is $0.0001$.

To train the neural network, we assume $\underline{a}= -1$, $\overline{H}=1$, i.e., the maximal investment is $1$ in each position\footnote{Note that in practice these bounds impose not a severe restriction as the resultant strategies can be scaled arbitrarily large if desired.}.

The training set of $45~000$ samples contains $34~146$ cases in which the market admits model-free arbitrage while the test set contains $3 787$ cases of  model-free arbitrage.

After training on the $45~000$ samples, the neural network assigns to  $41~768 $ out of $45~000$ the correct sign of the price of the strategy learned by the neural network, i.e., in $92.81\%$ of cases on the training set, the neural network can correctly decide whether the market admits arbitrage or not. On the test set we have $4460  $ out of $5000$ correct identifications which corresponds to 
$89.20  \%$. 
Compare also Figure~\ref{fig_training_progress} where we depict the loss function as well as the training and test set accuracy in dependence of the number of trained epochs.
\begin{figure}
\includegraphics[scale=0.55]{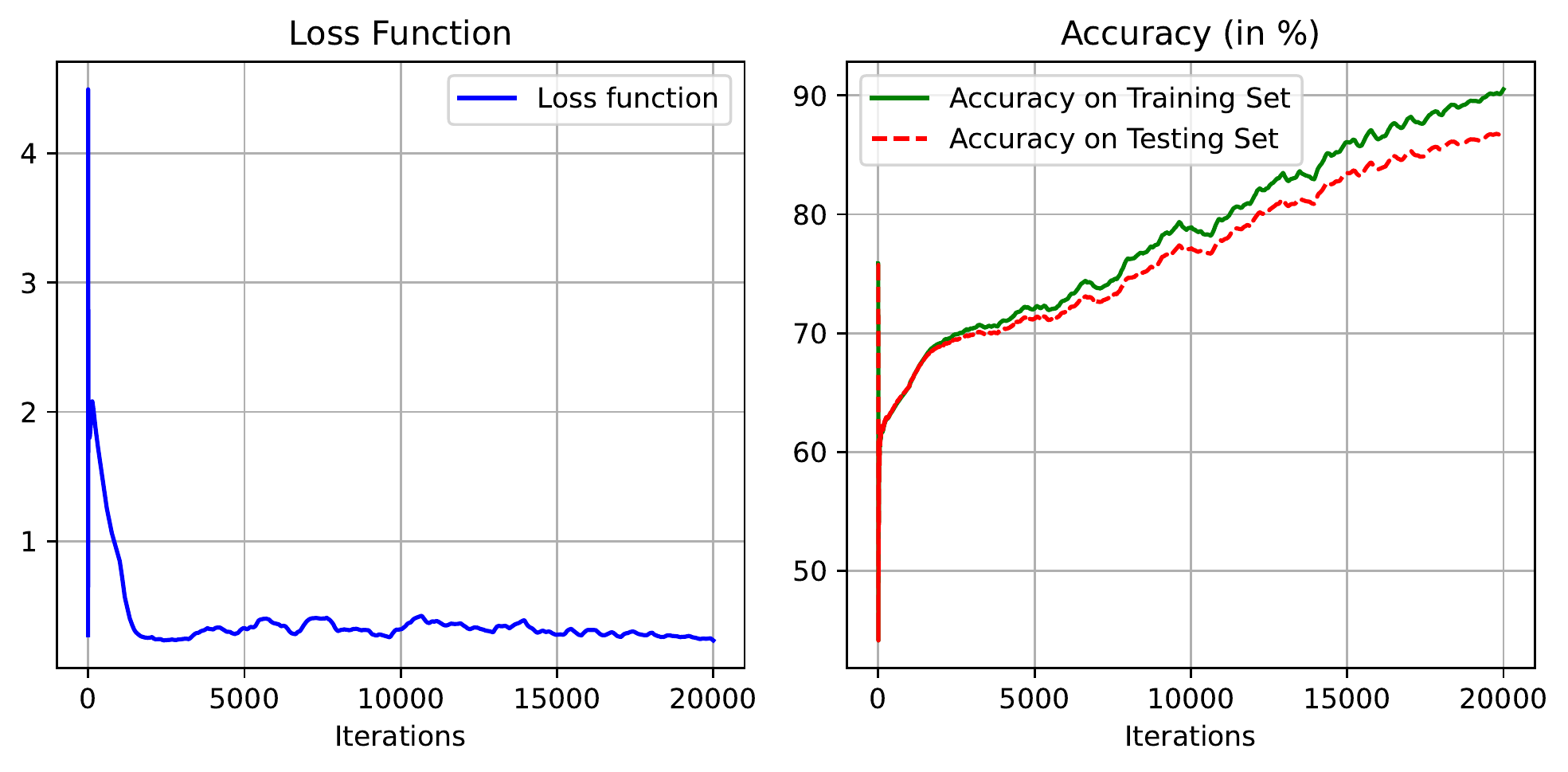}
\caption{The loss function as well as the training and test set accuracy in  dependenceof the number of trained epochs.}\label{fig_training_progress}
\end{figure}

However, it is important to emphasize that wrong identifications of the sign of the resultant strategy does not mean that the resultant strategy incur huge losses, as the magnitude of the predicted prices turns out to be on a small scale for the majority strategies with wrongly predicted sign. To showcase this, we evaluate the net profit $\mathcal{I}_{S_{i,j}}(K_i,a_i,h_i)-f(\pi_i,a_i,h_i)$ for $i = 1,\dots,5000$, $j =1,\dots, 200$, i.e., each of the $5000$ samples of the test set is evaluated on $200$ realizations of $S \in \mathcal{S}$ that are denoted by $S_{i,j}$ (uniformly sampled from $[0,2]^5$) and we show the results in Table~\ref{tbl_profit_test}. The results verify that on the test set the net profit is in the vast majority of the $1~000~000$ evaluated cases positive, compare also the histogram provided in Figure~\ref{fig_hist_test}. Moreover, the results support our claim that a wrong identification of arbitrage does not necessarily lead to a huge loss, as the net profits conditional on a wrong identification of existence of arbitrage turn out to be positive in most of the cases with a similar net profit distribution as in the unconditional case (right column of Table~\ref{tbl_profit_test} and right panel of Figure~\ref{fig_hist_test}).

\begin{table}[h!]
\begin{tabular}{lcc} \hline
       \toprule
       &\textbf{All samples} & \textbf{Conditional on wrong identification} \\
       \midrule
count 	&1~000~000  & 108~000\\
mean 	&0.477996& 0.468450\\
std 	&0.322428& 0.323814\\
min 	&-0.179760& -0.101957\\
25\% 	&0.215138& 0.202838\\
50\% 	&0.411158& 0.398024\\
75\% 	&0.705576& 0.700997\\
max 	&3.086741& 2.980903\\\hline 
\end{tabular}
\caption{Left column: The table shows the summary statistics of the net profit $\mathcal{I}_{S_{i,j}}(K_i,a_i,h_i)-f(\pi_i,a_i,h_i)$ for $i = 1,\dots,5000$, $j =1,\dots, 200$, i.e., each of the $5000$ samples of the test set is evaluated on $200$ realizations of $S \in \mathcal{S}$ leading to a total number of $1~000~000$ profits of the strategy trained as described in Section~\ref{sec_exa_1}. Right column: We depict the profit conditional on a wrong identification of existence of arbitrage opportunities.} \label{tbl_profit_test}
\end{table}

\begin{figure}[h!]
\begin{center}
\includegraphics[scale=0.4]{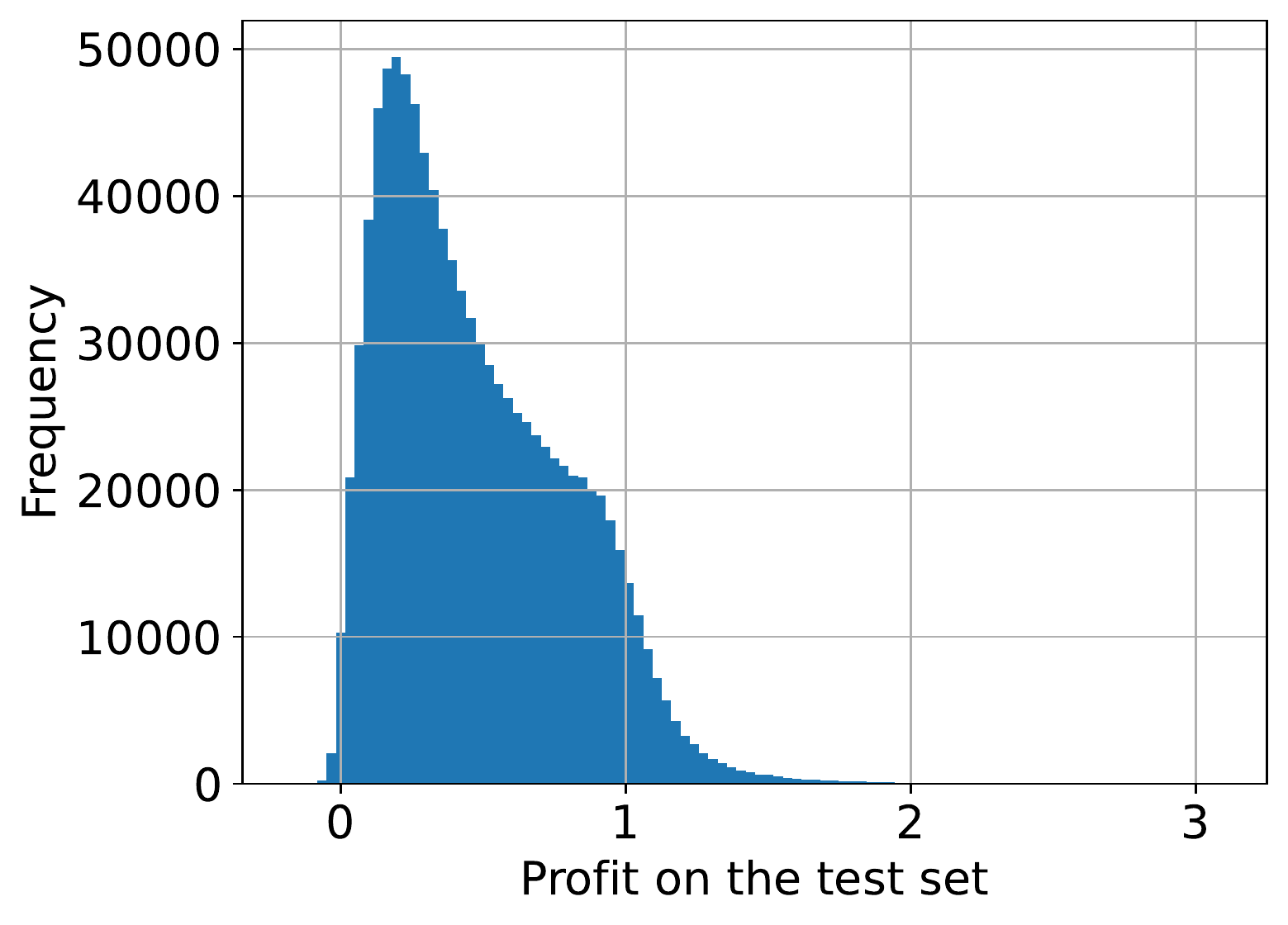}
\includegraphics[scale=0.4]{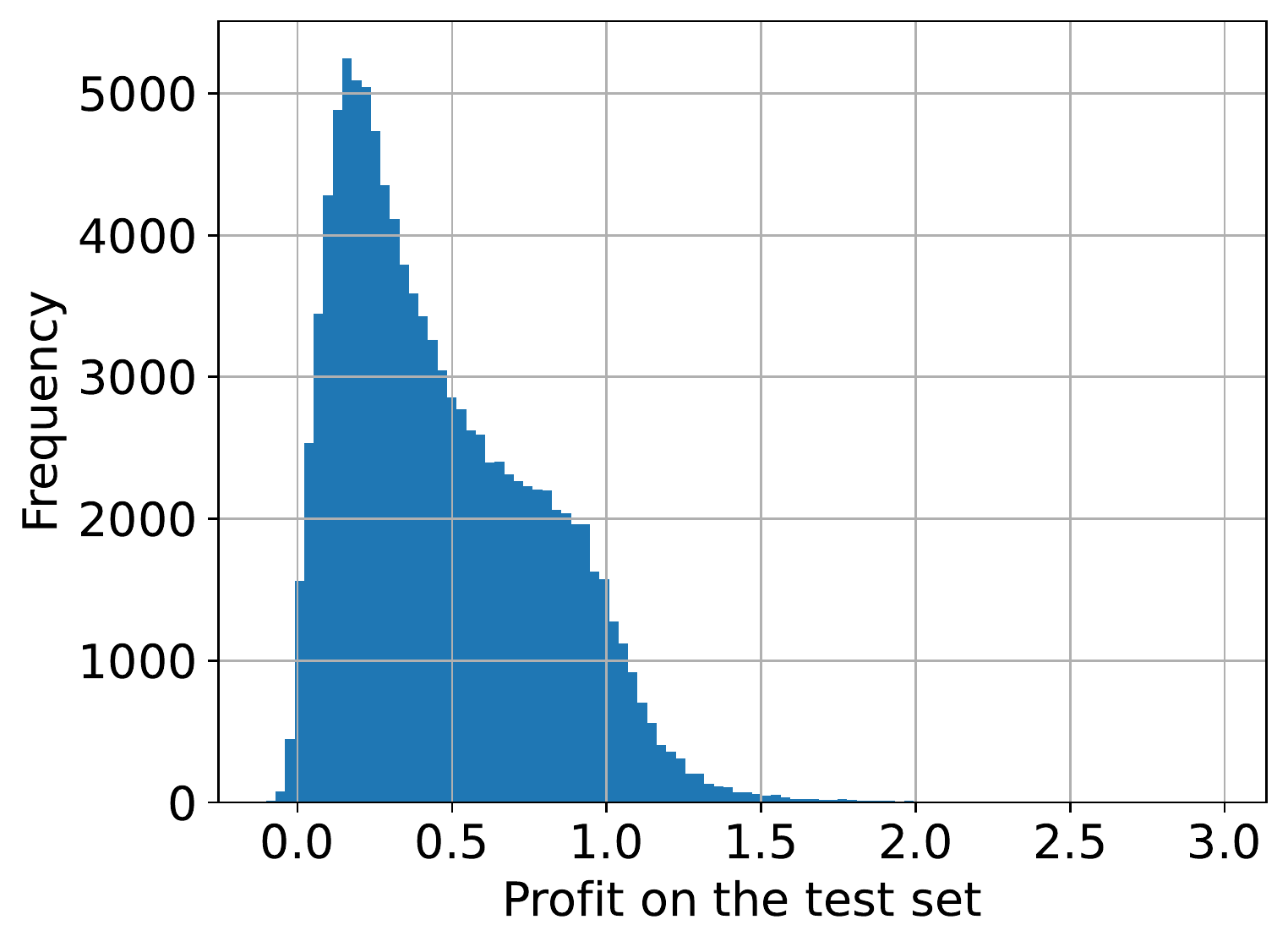}
\end{center}
\caption{Left: The histogram shows the distribution of the net profit $\mathcal{I}_{S_{i,j}}(K_i,a_i,h_i)-f(\pi_i,a_i,h_i)$ for $i = 1,\dots,5000$, $j =1,\dots, 200$ of the strategy trained as described in Section~\ref{sec_exa_1}. Right: The histogram shows the distribution of the net profit conditional on \emph{wrong identification} of arbitrage.}\label{fig_hist_test}
\end{figure}

\subsubsection{Stability of the results and choice of hyperparameters}~\\
We first remark that the results reported in Section~\ref{sec_exa_1} do not crucially depend on the ratio of the scenarios in which arbitrage can be observed (34~146 out of 45~000 samples). To show this, we first reduce the training set to a balanced training set containing $21~708$ samples of which $50\%$ constitute scenarios that allow for arbitrage. The results in Table~\ref{tbl_profit_balanced} and Figure~\ref{fig_hist_balanced} show that a strategy trained on this balanced training set according to Algorithm~\ref{algo_3} with the same hyperparameter as in Section~\ref{sec_exa_1}  performs even slightly better when being tested on the same (unbalanced) test set consisting of $5~000$ samples.

\begin{table}[h!]
\begin{tabular}{lc} \hline
 \\
count 	&1~000~000 \\
mean 	&0.694083\\
std 	&0.406774 \\
min 	&-0.136777 \\
25\% 	&0.390928 \\
50\% 	&0.653232 \\
75\% 	&0.936606\\
max 	&4.778378
\\\hline 
\end{tabular}
\caption{The table shows the summary statistics of the net profit $\mathcal{I}_{S_{i,j}}(K_i,a_i,h_i)-f(\pi_i,a_i,h_i)$ for $i = 1,\dots,5000$, $j =1,\dots, 200$, i.e., each of the $5000$ samples of the test set is evaluated on $200$ realizations of $S \in \mathcal{S}$ leading to a total number of $1~000~000$ profits of the strategy trained as described in Section~\ref{sec_exa_1} but on a reduced and balanced training set where 50 \% of the sampled constitute arbitrage situations.} \label{tbl_profit_balanced}
\end{table}

\begin{figure}[h!]
\begin{center}
\includegraphics[scale=0.4]{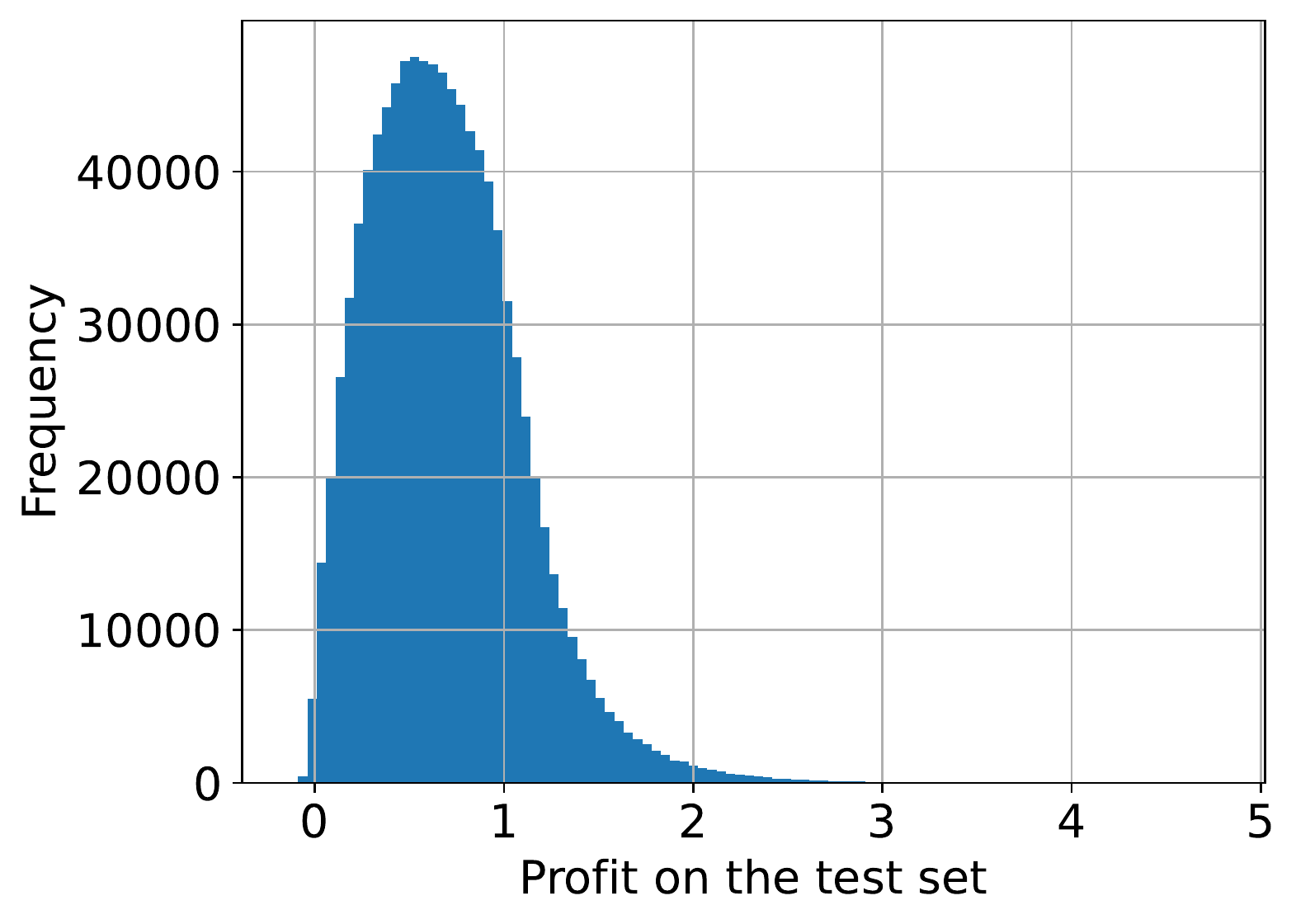}
\end{center}
\caption{The histogram shows the distribution of the net profit $\mathcal{I}_{S_{i,j}}(K_i,a_i,h_i)-f(\pi_i,a_i,h_i)$ for $i = 1,\dots,5000$, $j =1,\dots, 200$ of the strategy trained as described in Section~\ref{sec_exa_1} but on a reduced and balanced training set where 50 \% of the sampled constitute arbitrage situations.}\label{fig_hist_balanced}
\end{figure}

Next, we study the consequences of varying the used hyperparameters and we report in Table~\ref{tab:performance_metrics} the share of correct predictions, precision, recall, and F1 score for different configurations evaluated with respect to arbitrage detection on the test set. Moreover, Table~\ref{tab:summary_statistics_hyperparameters}  shows the performance of the profit of the trained strategies on the test set in dependence of different choices of hyperparameters. Taken together, the results reveal that the results do not vary significantly across different hyperparameter-configurations while a learning rate of $0.0001$, a depth of $5$ for the neural network, and no regularization seem to be close-to-optimal hyperparameter choices. Note also that the choice of a smaller learning rate of $0.001$ leads to a larger mean profit which however comes at the cost of a much higher standard deviation indicating a suboptimal solution. Moreover, increasing the depth of the neural network to $10$ prevents the neural network from learning a similar profitable strategy.

\begin{table}[h!]
    \centering
    \resizebox{\linewidth}{!}{ 
    \begin{tabular}{lcccc}
        \toprule
        \textbf{Learning Rate, Depth, Regularization} & \textbf{Correct Predictions} & \textbf{Precision} & \textbf{Recall} & \textbf{$F_1$ Score} \\
        \midrule
        0.0001,2,none & 0.8216 & 0.884871 & 0.878796 & 0.881823 \\
        0.0001,3,none & 0.8228 & 0.866195 & 0.905994 & 0.885648 \\
        0.0001,4,none & 0.8606 & 0.920065 & 0.893583 & 0.906631 \\
        0.0001,5,none & 0.8920 & 0.955911 & 0.898865 & 0.926511 \\
        0.0001,10,none & 0.2510 & 0.937500 & 0.011883 & 0.023468 \\
        0.001,5,none & 0.7186 & 0.811518 & 0.818590 & 0.815039 \\
        1e-05,5,none & 0.6936 & 0.866667 & 0.703723 & 0.776741 \\
        0.0001,5,l1 & 0.8936 & 0.959097 & 0.897808 & 0.927441 \\
        0.0001,5,l2 & 0.8186 & 0.815928 & 0.982044 & 0.891312 \\
        \bottomrule
    \end{tabular}}
    \caption{Performance metrics for different model parameters, evaluated on the test set.}
    \label{tab:performance_metrics}
\end{table}

\begin{table}[h!]
    \centering
    \resizebox{\linewidth}{!}{
    \begin{tabular}{lccccccccc}
        \toprule
        \textbf{Parameters} & \textbf{Count} & \textbf{Mean} & \textbf{Std} & \textbf{Min} & \textbf{25\%} & \textbf{50\%} & \textbf{75\%} & \textbf{Max} \\
        \midrule
        0.0001,2,none & 1000000 & 0.439912 & 0.297357 & -0.267825 & 0.203189 & 0.367542 & 0.644060 & 3.132133 \\
        0.0001,3,none & 1000000 & 0.436585 & 0.302734 & -0.209222 & 0.193420 & 0.360132 & 0.648741 & 3.113397 \\
        0.0001,4,none & 1000000 & 0.466350 & 0.310355 & -0.132638 & 0.217923 & 0.397229 & 0.681759 & 3.597595 \\
        0.0001,5,none & 1000000 & 0.477481 & 0.321754 & -0.157909 & 0.215212 & 0.410469 & 0.704891 & 3.220638 \\
        0.0001,10,none & 1000000 & 0.034884 & 0.025370 & -0.008109 & 0.015160 & 0.031994 & 0.050856 & 0.182104 \\
        0.001,5,none & 1000000 & 1.363648 & 0.837385 & -0.716115 & 0.733738 & 1.260230 & 1.896849 & 5.774780 \\
        1e-05,5,none & 1000000 & 0.406486 & 0.181477 & -0.189432 & 0.274503 & 0.392536 & 0.526439 & 1.299631 \\
        0.0001,5,l1 & 1000000 & 0.464646 & 0.317831 & -0.158561 & 0.204477 & 0.395615 & 0.690994 & 3.072829 \\
        0.0001,5,l2 & 1000000 & 0.468799 & 0.320785 & -0.174531 & 0.206054 & 0.399411 & 0.698279 & 3.331130 \\
        \bottomrule
    \end{tabular}
    \caption{Summary statistics for the net profit for strategies being trained using different hyperparameter combinations consisting of learning rate, depth of the neural network, and the type of regularization (None, $L_1$ or $L_2$), evaluated on the test set.}}
    \label{tab:summary_statistics_hyperparameters}
\end{table}

It is also noteworthy that the results are relatively robust with respect to the size of the training set on which the agent is trained. Indeed, in Table~\ref{tab:summary_statistics_samples} we report the profits on the test set for agents that have been trained on training sets of varying size. The results show that for all considered sample sizes between $5000$ and $45000$ the resultant profit on the test set is on a similar scale.
\begin{table}[h!]
    \centering
    \resizebox{\linewidth}{!}{
    \begin{tabular}{lccccccccc}
        \toprule
        \textbf{Training set size} & \textbf{5000} & \textbf{10000} & \textbf{15000} & \textbf{20000} & \textbf{25000} & \textbf{30000} & \textbf{35000} & \textbf{40000} & \textbf{45000} \\
        \midrule
        Count & 1000000 & 1000000 & 1000000 & 1000000 & 1000000 & 1000000 & 1000000 & 1000000 & 1000000 \\
        Mean & 0.546421 & 0.449859 & 0.459965 & 0.458827 & 0.466374 & 0.568733 & 0.538420 & 0.468137 & 0.477996 \\
        Std & 0.327194 & 0.304473 & 0.315945 & 0.309333 & 0.313317 & 0.247745 & 0.238015 & 0.321906 & 0.322428 \\
        Min & -0.175621 & -0.177209 & -0.214879 & -0.177739 & -0.188140 & -0.133989 & -0.112773 & -0.215613 & -0.179760 \\
        25\% & 0.292485 & 0.198430 & 0.201478 & 0.205287 & 0.210050 & 0.407968 & 0.386002 & 0.209109 & 0.215138 \\
        50\% & 0.493732 & 0.381197 & 0.389149 & 0.391185 & 0.398492 & 0.545277 & 0.514971 & 0.396287 & 0.411158 \\
        75\% & 0.763128 & 0.672535 & 0.686655 & 0.681570 & 0.691300 & 0.694669 & 0.655317 & 0.691226 & 0.705576 \\
        Max & 3.851340 & 2.628088 & 3.276558 & 3.169339 & 3.166770 & 4.146131 & 3.460321 & 4.275422 & 3.086741 \\
        \bottomrule
    \end{tabular}}
    \caption{Summary statistics for different sample sizes}
    \label{tab:summary_statistics_samples}
\end{table}
 
Moreover, even though the optimization procedure is random due to the use of random initialization of weights and the use of the Adam optimizer, the resultant profits on the test turn out to be on the same scale when being trained several times with the same hyperparameter configuration, compare also Table~\ref{tab:summary_statistics_parameters_several_runs} where we show the net profit for the optimization procedure as specified in Section~\ref{sec_exa_1}  trained $10$ times independently. 
\begin{table}[h!]
    \centering
    \resizebox{\linewidth}{!}{
    \begin{tabular}{lcccccccccc}
        \toprule
        Count & 1000000 & 1000000 & 1000000 & 1000000 & 1000000 & 1000000 & 1000000 & 1000000 & 1000000 & 1000000 \\
        Mean & 0.452637 & 0.472013 & 0.442239 & 0.474521 & 0.485375 & 0.449283 & 0.437489 & 0.464786 & 0.467977 & 0.461433 \\
        Std & 0.309939 & 0.323847 & 0.306825 & 0.312429 & 0.327781 & 0.305616 & 0.311088 & 0.322367 & 0.309843 & 0.312393 \\
        Min & -0.215139 & -0.163364 & -0.163844 & -0.241484 & -0.172754 & -0.196785 & -0.128142 & -0.127955 & -0.162520 & -0.159971 \\
        25\% & 0.199371 & 0.211284 & 0.190744 & 0.227038 & 0.227661 & 0.199840 & 0.188147 & 0.210150 & 0.217428 & 0.205545 \\
        50\% & 0.380630 & 0.399340 & 0.369056 & 0.405421 & 0.413370 & 0.378361 & 0.354790 & 0.389295 & 0.399316 & 0.392809 \\
        75\% & 0.675014 & 0.697409 & 0.663718 & 0.688814 & 0.705762 & 0.668901 & 0.653276 & 0.680747 & 0.687475 & 0.685806 \\
        Max & 3.252239 & 3.669023 & 2.746479 & 5.333844 & 4.395148 & 2.827598 & 3.379474 & 4.312405 & 3.028918 & 2.850765 \\
        \bottomrule
    \end{tabular}}
    \caption{Summary statistics of the net profit on the test set for several strategies trained with the same hyperparameters.}
    \label{tab:summary_statistics_parameters_several_runs}
\end{table}

\subsubsection{Backtesting with historical option prices }\label{sec_exa_2}
We backtest the strategy trained in Section~\ref{sec_exa_1} on the stocks of \emph{Apple}, \emph{Alphabet}, \emph{Microsoft}, \emph{Google}, and \emph{Meta}. To this end, we consider for each of the companies call options with maturity $24$ March $2023$ for ten different strikes.

The bid and ask prices of these call options and the underlying securities were observed on $33$ trading days ranging from $2$ February $2023$ until $22$ March $2023$.

We apply the strategy trained in Section~\ref{sec_exa_1} to the prices observed on each of the $33$ trading days and evaluate it on the realized values of the $5$ underlying securities at maturity. In Table~\ref{tbl_real_unscaled} and Figure~\ref{fig_real_unscaled} we summarize the net profits of the $33$ strategies. Note that to apply the trained neural network from Section~\ref{sec_exa_1}, we first scale all the financial instruments such that the spot values of the underlying securities equal 1, as described in Section~\ref{sec_exa_1}. Then, after applying the strategies to the scaled inputs, we rescale the values of the involved quantities back to unnormalized values, and we  report in Table~\ref{tbl_real_unscaled} and Figure~\ref{fig_real_unscaled} the net profits for both cases: after rescaling the values of the underlying securities, options, and strikes to unnormalized values, as well as without scaling back.
The results of the backtesting study reveal that even though the neural network from Section~\ref{sec_exa_1} was trained on data extracted at a different day ($25$ April $2023$) involving call options with a different maturity written on other assets, the resultant strategy still allows to trade profitably in the majority of cases, showcasing the robustness of our algorithm.

\begin{table}[h!]
\begin{tabular}{lcc} 
 & \textbf{unscaled} &\textbf{scaled} \\  \hline
count 	&33 &33 \\
mean 	&5.855233 &0.063460\\
std 	&9.706270 &0.089347\\
min 	&-4.225227 &-0.017063\\
25\% 	&-1.475227 &-0.007025\\
50\% 	&1.535172 &0.021668\\
75\% 	&13.032318 &0.093558\\
max 	&32.687988 &0.318201\\\hline 
\end{tabular}
\caption{In the setting of Section~\ref{sec_exa_2}, the table shows the summary statistics of the net profit $\mathcal{I}_{S_T}(K_i,a_i,h_i)-f(\pi_i,a_i,h_i)$ for $i = 1,\dots,33$, where here $S_T\in \R^5$ refers to the observed realization of the  $5$ underlying securities at maturity $T=$ $24$ March $2023$. To apply the trained neural network we first scale the values such that the spot prices of the underlying assets equal 1. The left column shows the values after scaling the values back, whereas the right column shows the statistics directly after applying the neural network to the scaled data.} \label{tbl_real_unscaled}
\end{table}

\begin{figure}[h!]
\begin{center}
\includegraphics[scale=0.4]{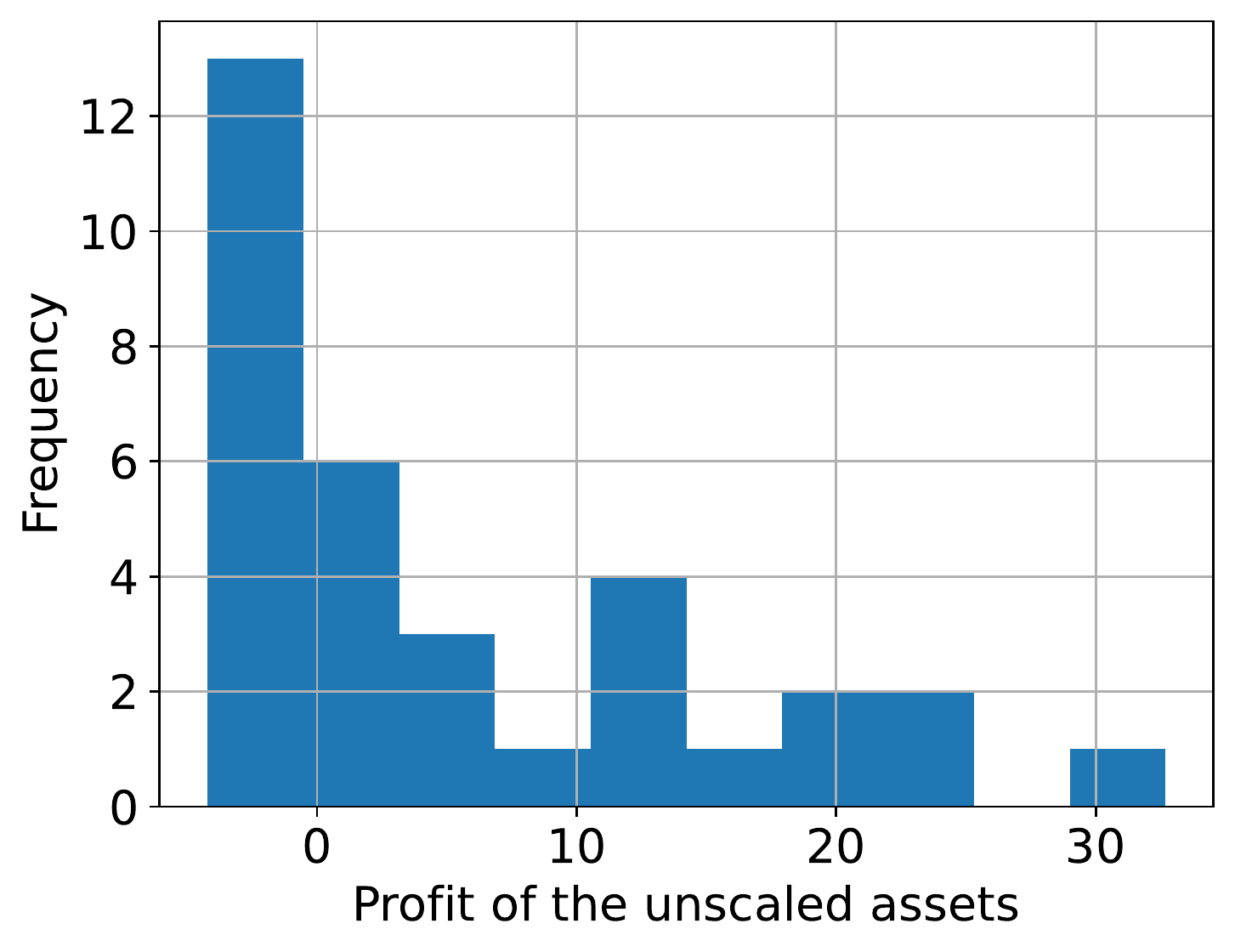}
\includegraphics[scale=0.4]{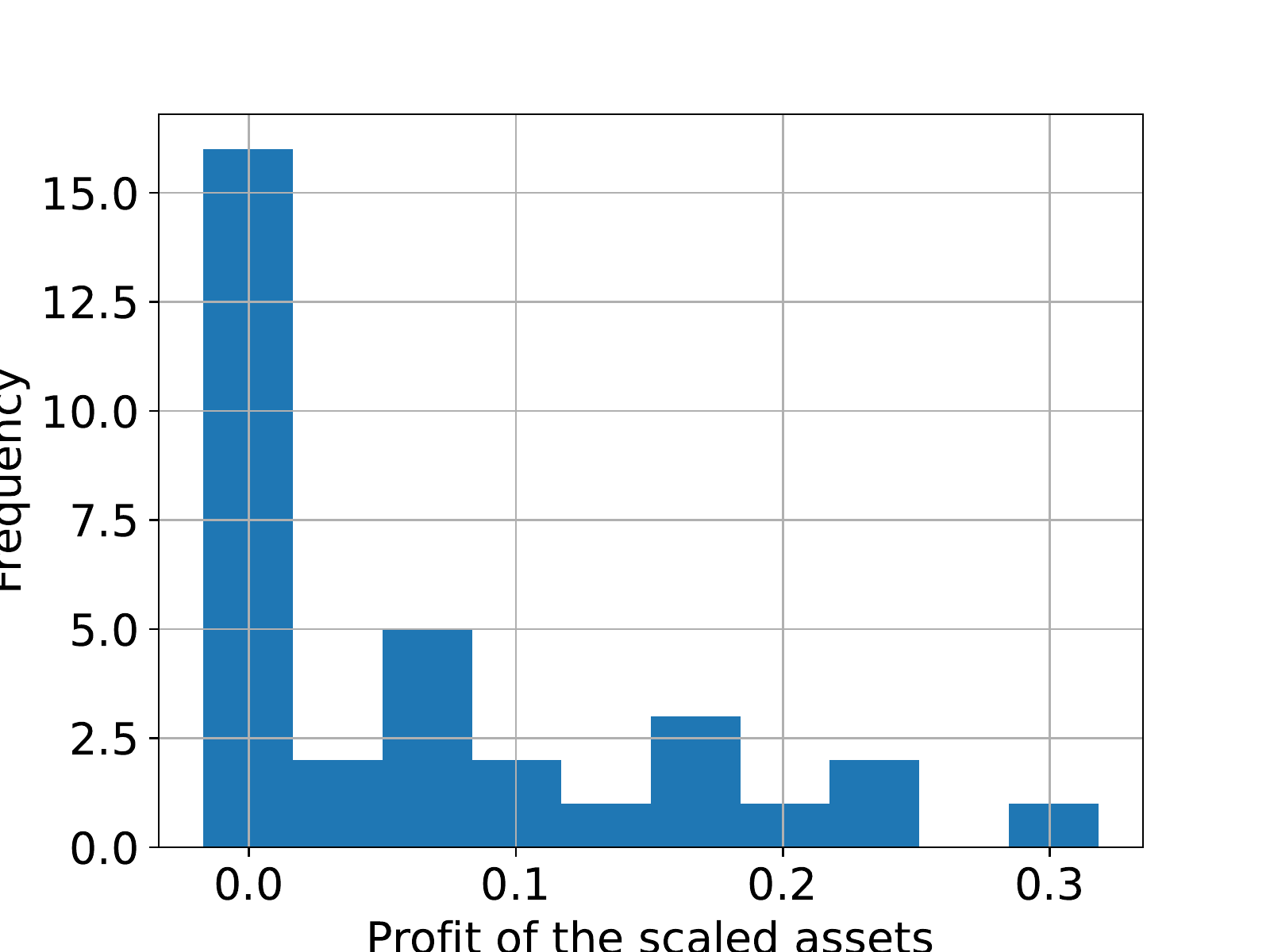}
\end{center}
\caption{In the setting of Section~\ref{sec_exa_2}, the histogram depicts the net profits $\mathcal{I}_{S_T}(K_i,a_i,h_i)-f(\pi_i,a_i,h_i)$ for $i = 1,\dots,33$, where here $S_T\in \R^5$ refers to the observed realization of the $5$ underlying securities at maturity $T=$ $24$ March $2023$.}\label{fig_real_unscaled}
\end{figure}

\newpage
\section{Approximation of optimal solutions of general convex semi-infinite programs by neural networks}\label{sec_convex}
In this section we show for a certain class of convex semi-infinite optimization problems (CSIP) that each of them can be approximately solved by a \textit{single} neural network. More precisely, for every prescribed accuracy $\varepsilon>0$ we show that there exists a \textit{single} neural network which outputs a \textit{feasible} solution which is $\varepsilon$-optimal.
This class of convex semi-infinite problems covers the setting of static arbitrage detection introduced in Section~\ref{sec_arbitrage} as special case. We leave further applications for future research.
%
\subsection{Setting}

Let $\underline{a}\in \R$, let $\K_x \subset \R^{n_x}$ be compact for some $n_x \in \N$, and let $\K_y \subset \R^{n_y}$ be compact and convex for some $n_y \in \N$.

We consider some function 
\[
f: \K_x \times [\underline{a},\infty) \times \K_y \ni (x,a,y) \mapsto f(x,a,y)\in \R,
\]
which we aim to minimize under suitable constraints. To define these constraints we consider  some (possibly uncountable infinite) index set $\mathcal{S}$  as well as for all $s \in \mathcal{S}$ a function 
\[
\K_x \times [\underline{a},\infty) \times \K_y \ni (x,a,y) \mapsto \mathcal{I}_s(x,a,y) \in \R.
\]
Further, let $\K_x \ni x \twoheadrightarrow \Gamma(x)\subseteq [\underline{a},\infty) \times \K_y$ be the correspondence defined by
\[
\Gamma(x):= \left\{ (a,y) \in [\underline{a},\infty) \times \K_y ~\middle|~ -\mathcal{I}_s(x,a,y) \leq 0 \text{ for all } s \in \mathcal{S}\right\},
\]
that defines the set of \emph{feasible} elements from $[\underline{a},\infty) \times \K_y$. To define our optimization problem, we now consider the function $\K_x \ni x \mapsto V(x) \in \R$ defined by
\begin{equation}\label{eq:CSIP}
V(x):= \inf_{(a,y) \in \Gamma(x)} f(x,a,y).
\end{equation}
We impose the following assumptions on the above defined quantities.
\begin{asu}[Assumptions on $f$]\label{asu_f}~
\begin{itemize}
\item[(i)] There exists some $L_f\geq 1$ such that the function $[\underline{a},\infty) \times \K_y \ni (a,y) \mapsto f(x,a,y)$ is $L_f$-Lipschitz continuous for all $x \in \K_x$.
\item[(ii)] The function $ \K_x \times [\underline{a},\infty) \times \K_y \ni(x,a,y) \mapsto f(x,a,y)$ is continuous.
\item[(iii)]  The function $[\underline{a},\infty)\times \K_y \ni (a,y) \mapsto f(x,a,y)$ is convex for all $x \in \K_x$.
\item[(iv)] The function $[\underline{a},\infty)\ni a \mapsto f(x,a,y)$ is increasing for all $x \in \K_x$, $y\in \K_y$.
\item[(v)] We have that 
\[
\mathcal{L}_{a,f}:= \inf_{x \in \K_x, y \in \K_y \atop a_1,a_2 \in [\underline{a},\infty), a_1 \neq a_2} \frac{|f(x,a_1,y)-f(x,a_2,y)|}{|a_1-a_2|}>0.
\]
\end{itemize}
\end{asu}
\begin{asu}[Assumptions on $\mathcal{I}_s$]\label{asu_IS}~
\begin{itemize}
\item[(i)] There exists some ${L}_{\mathcal{I}} \geq 1$ such that $\K_x \times [\underline{a},\infty) \times \K_y \ni (x,a,y) \mapsto \mathcal{I}_{s}(x,a,y)$ is ${L}_{\mathcal{I}}$-Lipschitz continuous for all $s \in \mathcal{S}$.
\item[(ii)]The function $[\underline{a},\infty) \times \K_y  \ni (a,y) \mapsto \mathcal{I}_s(x,a,y)$ is concave for all $x \in \K_x, s \in \mathcal{S}$.
\item[(iii)] The function $[\underline{a},\infty) \ni a \mapsto \mathcal{I}_s(x,a,y)$ is increasing for all $x \in \K_x, y \in \K_y, s \in \mathcal{S}$.
\item[(iv)] We have that 
\[
\mathcal{L}_{a,\mathcal{I}}:=\inf_{s \in \mathcal{S}} \inf_{x \in \K_x, \atop  y \in \K_y} \inf_{a_1 \neq a_2, \atop a_1,a_2 \in [\underline{a},\infty)} \frac{\left|\mathcal{I}_{s}(x,a_1,y)-\mathcal{I}_{s}(x,a_2,y)\right|}{|a_1-a_2|} >0.
\]
\item[(v)] We have that  $$\inf_{s \in \mathcal{S}, \atop x \in \K_x, y \in \K_y} \mathcal{I}_s(x,\underline{a},y)>-\infty.$$

\end{itemize}
\end{asu}
\begin{asu}[Assumptions on $\K_y$]\label{asu_KxKy}~
There exists $0 < r<1$ and $ L_r \geq 1$ such that for all $0 < \delta< r$ there exists some closed and convex  set $C_{y,\delta} \subset \K_y$ such that for all $y' \in C_{y,\delta} $, $y \in \R^{n_y}$ we have 
\[
\|y'-y\| \leq \delta \Rightarrow y\in  \K_y,
\]
and 
\[
\max_{y \in \K_y} \min_{y' \in C_{y,\delta}} \left\{\|y-y'\|\right\} \leq L_r \delta.
\]
\end{asu}

\begin{rem}[On the assumptions]\label{rem_assumptions}~
\begin{itemize}

\item[(i)]
Let 
\begin{equation}\label{eq_defn_a_UB}
\overline{a}^{\operatorname{UB}}:=\underline{a}+\frac{1}{\mathcal{L}_{a,\mathcal{I}}}\left|\inf_{s \in \mathcal{S}, \atop x \in \K_x, y \in \K_y} \mathcal{I}_s(x,\underline{a},y)\right|\in [\underline{a},\infty)
\end{equation}
Then, we have $\mathcal{I}_s(x,\overline{a}^{\operatorname{UB}},y) \geq 0$ for all $x \in \K_x$, $y \in \K_y$, $s \in \mathcal{S}$. In particular, $\Gamma(x) \neq \emptyset$ for all $x \in \K_x$. Indeed, by using the definition of $\overline{a}^{\operatorname{UB}}$ and $ \mathcal{L}_{a,\mathcal{I}} $ together with Assumption~\ref{asu_IS}~(v) we have for all $x\in \K_x$, $y \in \K_y$, $s \in \mathcal{S}$ that
\begin{align*}
\mathcal{I}
_s(x,\overline{a}^{\operatorname{UB}},y)&=\mathcal{I}
_s(x,\overline{a}^{\operatorname{UB}},y)-\mathcal{I}
_s(x,\underline{a},y)+\mathcal{I}
_s(x,\underline{a},y)\\
&\geq \mathcal{L}_{a,\mathcal{I}} \cdot \overline{a}^{\operatorname{UB}}-\mathcal{L}_{a,\mathcal{I}}  \cdot \underline{a} + \inf_{s \in \mathcal{S}, \atop x \in \K_x, y \in \K_y}\mathcal{I}
_s(x,\underline{a},y) \\
&= \mathcal{L}_{a,\mathcal{I}}  \cdot \left( \underline{a}+\frac{1}{\mathcal{L}_{a,\mathcal{I}}} \cdot \left|\inf_{s \in \mathcal{S}, \atop x \in \K_x, y \in \K_y}\mathcal{I}
_s(x,\underline{a},y)\right|\right)-\mathcal{L}_{a,\mathcal{I}}  \cdot \underline{a} + \inf_{s \in \mathcal{S}, \atop x \in \K_x, y \in \K_y}\mathcal{I}
_s(x,\underline{a},y) \geq  0.
\end{align*}
\item[(ii)]Assumption~\ref{asu_f}~(ii) and ~(iv), and the assumption that  $\K_x$ and $\K_y$ are compact ensure together with Remark~\ref{rem_assumptions}~(i) that $V(x) \in \R$ for all $x \in \K_x$. Indeed, for any $x \in \K_x$ and $(a,y) \in \Gamma(x)$, we have $f(x,a,y) \geq f(x,\underline{a},y)\geq \inf_{x\in \K_x, y \in \K_y}f(x,\underline{a},y)>-\infty$.
\item[(iii)] Assumption~\ref{asu_f}~(iv) and ~(v) ensure that the function $f$ is strictly increasing in $a\in[\underline{a},\infty)$ uniformly in $x\in \K_x$, $y\in\K_y$. Analogously, Assumption~\ref{asu_IS}~(iii) and ~(iv) ensure that the function $\mathcal{I}$ is strictly increasing in $a\in[\underline{a},\infty)$ uniformly in $x\in \K_x$, $y\in\K_y$, $s\in \mathcal{S}$.
\item[(iv)] Note that Assumption~\ref{asu_KxKy} roughly speaking means that the geometry of $\K_y\subseteq \R^{n_y}$ is similar to a box. Indeed, if $\K_y=\times_{i=1}^{n_y} [l_i,u_i]$ for some $-\infty<\l_i<u_i<\infty$, $i=1,\dots,n_y$, then one can choose $0<r<\min_{i} (\frac{u_i-l_i}{2})>0$, $C_{y,\delta}:=\times_{i=1}^{n_y} [l_i+\delta,u_i-\delta]\subseteq \K_y$, and $L_r:=\sqrt{n_y}$. 
\end{itemize}
\end{rem}
Our main result of this section establishes the existence of a \textit{single} neural network such that for any input $x\in \K_x$ defining the (CSIP) in \eqref{eq:CSIP} the neural network outputs a \textit{feasible} solution which is \textit{$\varepsilon$-optimal}.

\begin{thm}[Single neural network provides corresponding feasible $\varepsilon$-optimizer for class of~(CSIP)]\label{thm_main}
Let Assumptions~\ref{asu_f}, \ref{asu_IS}, and \ref{asu_KxKy} hold true. Then, for all $\varepsilon>0$ there exists a neural network $\NN \in  \mathfrak{N}_{n_x,1+n_y}$ such that 
\begin{itemize}
\item[(i)] $\NN(x) := \left(\NN_a(x), \NN_y(x)\right) \in \Gamma(x)$ for all $x \in \K_x$,
\item[(ii)] $f\left(x, \NN_a(x), \NN_y(x)\right) - V(x) \leq \varepsilon$ for all $x \in \K_x$.
\end{itemize}
\end{thm}
\vspace{0.1cm}
\noindent
The proof of Theorem~\ref{thm_main} is provided in the next section. 

\section{Proofs and Auxiliary Results}\label{sec_proofs}
In this section, we present the proofs of the main results from Section~\ref{sec_arbitrage} and ~\ref{sec_convex}.
\subsection{Proofs of Section~\ref{sec_arbitrage}}
The proof of Proposition~\ref{prop_arbitrage} consists of verifying that the optimization problem~\eqref{eq_defn_Psi_arbitrage_setting} is included in the general (CSIP) introduced in Section~\ref{sec_convex}. 
Then, applying Proposition~\ref{prop_arbitrage} together with the universal approximation property of neural networks allows to conclude Theorem~\ref{cor_ftpa} and Theorem~\ref{thm:epsilon-arbitrage}.
\begin{proof}[Proof of Proposition~\ref{prop_arbitrage}]
We verify that the conditions imposed in Theorem~\ref{thm_main} are satisfied under Assumption~\ref{asu_arbitrage} with $x \leftarrow (K,\pi)$, $a \leftarrow a$, $y \leftarrow h$, $V \leftarrow V$ in the notation of Theorem~\ref{thm_main}. To that end, note that Assumption~\ref{asu_f} holds with $\mathcal{L}_{a, f}=1$, and $L_f = \max\{1,~\overline{\pi}\}\sqrt{1+2N}$. Moreover, note that for all $x \in \K_x$, $(a,0) \in \left([\underline{a},\infty) \cap [0,\infty)\right)\times [0,\overline{H}]^{2N}$ satisfies
$\mathcal{I}_S(x,a,0) \geq 0$ for all $S \in \mathcal{S}$. Hence, Assumption~\ref{asu_IS} holds with $\mathcal{L}_{a, \mathcal{I}} = 1$ and $L_{\mathcal{I}} = \max\{1,~2\overline{H}L_\Psi,~C_{\Psi}\}\sqrt{3N+1}$. Furthermore, for any $0<r<1$ and any $0 < \delta<1$, Assumption~\ref{asu_KxKy} is satisfied with $C_{y,\delta}=[\delta,\overline{H}-\delta]^{2N}$ and $L_r=\sqrt{2N}$. Therefore, the result follows by Theorem~\ref{thm_main}.

\end{proof}

\begin{proof}[Proof of Theorem~\ref{cor_ftpa}]
Let $(K,\pi)\in [0, \overline{K}]^{N}\times [0, \overline{H}]^{2N}$.
Assume first there exists \emph{model-free arbitrage}, i.e., we have $V(K,\pi)<0$. Then, we choose $\varepsilon \in \R$ with  $0<\varepsilon<-V(K,\pi)$ and obtain with Proposition~\ref{prop_arbitrage} the existence of a neural network $ \NN =(\NN_a,\NN_h) \in  \mathfrak{N}_{3N,1+2N} $ with $\NN(K,\pi) \in \Gamma(K)$ and with 
 $f\left(\pi,\NN_a(K,\pi),\NN_h(K,\pi)\right)-V(K,\pi) \leq \varepsilon$ which implies
 \[
 f\left(\pi,\NN_a(K,\pi),\NN_h(K,\pi)\right) \leq \varepsilon+V(K,\pi) <-V(K,\pi)+V(K,\pi) =0.
 \]
Conversely, if conditions (i) and (ii) hold, then the output of the neural network constitutes a \emph{model-free arbitrage} opportunity.
\end{proof}
%
%
\begin{proof}[Proof of Theorem~\ref{thm:epsilon-arbitrage}]
Let $\varepsilon>0$. By Proposition~\ref{prop_arbitrage}, there exists a neural network  $\NN \in  \mathfrak{N}_{3N,1+2N} $ such that for every $(K,\pi) \in [0, \overline{K}]^{N} \times [0, \overline{\pi}]^{2N}$
\begin{equation*}
\begin{aligned}
&\NN(K,\pi):=\left(\NN_a(K,\pi),\NN_h(K,\pi)\right) \in \Gamma(K) \\
 \mbox{ and } \
 &f\left(\pi,\NN_a(K,\pi),\NN_h(K,\pi)\right)-V(K,\pi) \leq\varepsilon- \delta.
\end{aligned}
\end{equation*}
Moreover, for every $(K,\pi) \in [0, \overline{K}]^{N} \times [0, \overline{\pi}]^{2N}$, if the  market with respect to $(K,\pi)$ admits model-free static arbitrage of magnitude $\varepsilon$, then by definition $V(K,\pi)\leq -\varepsilon$. This  implies that $\NN(K,\pi):=\left(\NN_a(K,\pi),\NN_h(K,\pi)\right)$ provides a model-free static arbitrage strategy of magnitude $\delta$.

If the  market with respect to $(K,\pi)$ admits no model-free static arbitrage, then $V(K,\pi)=0$ and hence 
\[
f\left(\pi,\NN_a(K,\pi),\NN_h(K,\pi)\right) \leq V(K,\pi)+\varepsilon- \delta =\varepsilon-\delta.
\] 
\end{proof}
%
\noindent
It remains to prove Theorem~\ref{thm_main}, which is our main \textit{technical} result. Its proof is provided in the next subsection.
\subsection{Proofs of Section~\ref{sec_convex}}
The main idea of the proof of Theorem~\ref{thm_main} is to show that the correspondence of feasible $\varepsilon$-optimizers of the convex semi-infinite program (CSIP) defined in \eqref{eq:CSIP}, as a function of the input $x\in \K_x$ of the (CSIP),  is  
non-empty, convex, closed, and lower hemicontinuous\footnote{We refer to, e.g., \cite[Chapter~17]{Aliprantis} as reference for the standard notions of lower/upper (hemi)continuity of correspondences.}, where the major difficulty lies in the establishment of the lower hemicontinuity. This then allows us to apply Michael's continuous selection theorem (\cite{michael}), which together with the universal approximation property of neural networks leads to the existence of a \textit{single} neural network which for any input $x\in \K_x$ defining the (CSIP) in \eqref{eq:CSIP} outputs a \textit{feasible} solution which is $\varepsilon$-optimal. We highlight that no strict-convexity of the map $(a,y)\mapsto f(x,a,y)$ for any fixed $x$ is assumed in \eqref{eq:CSIP}, hence one cannot expect uniqueness of optimizers for the (CSIP), which in turn means that one cannot expect to have lower hemicontinuity of the correspondence of feasible true optimizers of the (CSIP) in \eqref{eq:CSIP}. 

\subsubsection{Auxiliary Results}\label{sec_auxiliary}
Before reporting the proof of Theorem~\ref{thm_main}, we establish several auxiliary results which are necessary for the proof of the main result from Theorem~\ref{thm_main}. 

For all of the auxiliary results from Section~\ref{sec_auxiliary} we assume the validity of Assumption~\ref{asu_f}, Assumption~\ref{asu_IS} and Assumption~\ref{asu_KxKy}. Moreover, from now on, we define the following quantity

\begin{equation}\label{eq_defn_a_UB_2}
\overline{a}^{\operatorname{UB}}:=\underline{a}+\frac{1}{\mathcal{L}_{a,\mathcal{I}}}\left|\inf_{s \in \mathcal{S}, \atop x \in \K_x, y \in \K_y} \mathcal{I}_s(x,\underline{a},y)\right|\in [\underline{a},\infty).
\end{equation}

\begin{lem} \label{lem_convex_1}~
\begin{itemize}
\item[(i)]
Let $a \in [\underline{a},\infty)$ such that $a\geq \overline{a}^{\operatorname{UB}}$.
Then, we have that $\mathcal{I}_s(x,a,y) \geq 0 $ for all $x \in \K_x$, $y \in \K_y$, $s \in \mathcal{S}$.
\item[(ii)] Let $a \in [\underline{a},\infty)$ such that $a \geq \overline{a}^{\operatorname{UB}}+\frac{1}{\mathcal{L}_{a,f}}$. Then, for all $x \in \K_x$ and for all $y \in \K_y$ we have $f(x,a,y)-V(x) \geq 1$.
\end{itemize}
\end{lem}
\begin{proof}~
\begin{itemize}
\item[(i)]
Let $a \in [\underline{a},\overline{a}]$ such that $a\geq \overline{a}^{\operatorname{UB}}$. Further, let $x\in \K_x$, $y \in \K_y$, $s \in \mathcal{S}$. Then, we have by the  monotonicity of $\mathcal{I}_s$ on $[\underline{a},\overline{a}]$ (stated in Assumption~\ref{asu_IS}~(iii)) that
\begin{equation}\label{eq_proof_51_first_ineq}
\mathcal{I}_s(x,a,y) \geq \mathcal{I}_s(x,\overline{a}^{\operatorname{UB}},y) = \mathcal{I}_s(x,\overline{a}^{\operatorname{UB}},y)-\mathcal{I}_s(x,\underline{a},y)+\mathcal{I}_s(x,\underline{a},y).
\end{equation}
By using the above inequality \eqref{eq_proof_51_first_ineq}, Assumption~\ref{asu_IS}~(iv), and the definition of $\overline{a}^{\operatorname{UB}}$ we then have 
\begin{align*}
\mathcal{I}_s(x,a,y) &\geq \mathcal{L}_{a,\mathcal{I}} \cdot \left(\overline{a}^{\operatorname{UB}}-\underline{a}\right)+\inf_{s \in \mathcal{S}, \atop x \in \K_x, y \in \K_y} \mathcal{I}_s(x,\underline{a},y)\\
&=\mathcal{L}_{a,\mathcal{I}} \cdot \overline{a}^{\operatorname{UB}}-\mathcal{L}_{a,\mathcal{I}}\cdot \underline{a}+\inf_{s \in \mathcal{S}, \atop x \in \K_x, y \in \K_y} \mathcal{I}_s(x,\underline{a},y)\\
&=\mathcal{L}_{a,\mathcal{I}} \cdot\left(\underline{a}+\frac{1}{\mathcal{L}_{a,\mathcal{I}}} \left| \inf_{s \in \mathcal{S}, \atop x \in \K_x, y \in \K_y} \mathcal{I}_s(x,\underline{a},y)\right| \right)-\mathcal{L}_{a,\mathcal{I}}\cdot \underline{a}+\inf_{s \in \mathcal{S}, \atop x \in \K_x, y \in \K_y} \mathcal{I}_s(x,\underline{a},y) \\
&= \left| \inf_{s \in \mathcal{S}, \atop x \in \K_x, y \in \K_y} \mathcal{I}_s(x,\underline{a},y) \right| +  \inf_{s \in \mathcal{S}, \atop x \in \K_x, y \in \K_y} \mathcal{I}_s(x,\underline{a},y) \geq 0.
\end{align*}
\item[(ii)] First note that, by the assertion from (i), we have $(\overline{a}^{\operatorname{UB}},y) \in \Gamma(x)$ for all $y \in \K_y$ and hence \begin{equation} \label{eq_fgeqpsi}
f(x,\overline{a}^{\operatorname{UB}},y) \geq V(x) \text{ for all } x \in \K_x,~y \in \K_y.
\end{equation}
Then, as by assumption $a \geq \overline{a}^{\operatorname{UB}}+\frac{1}{\mathcal{L}_{a,f}}$, we have for all $x \in \K_x$ and for all $y \in \K_y$ by Assumption~\ref{asu_f}~(iv), by Assumption~\ref{asu_f}~(v), and by  \eqref{eq_fgeqpsi}, that
\begin{align*}
f(x,a,y)-V(x) &=f(x,a,y) -f(x,\overline{a}^{\operatorname{UB}},y)+f(x,\overline{a}^{\operatorname{UB}},y)-V(x) \\
& \geq \mathcal{L}_{a,f} \cdot (a- \overline{a}^{\operatorname{UB}})+f(x,\overline{a}^{\operatorname{UB}},y)-V(x)\\
&\geq \mathcal{L}_{a,f}  \cdot (a- \overline{a}^{\operatorname{UB}}) \geq 1.
\end{align*}
\end{itemize}
\end{proof}

From now on, let 
\begin{equation}\label{eq_defn_overline_a}
\overline{a}:= \overline{a}^{\operatorname{UB}}+\frac{1}{\mathcal{L}_{a,f}}+2,
\end{equation}
where $\overline{a}^{\operatorname{UB}}$ is defined in \eqref{eq_defn_a_UB_2}. Moreover, we define the correspondence
\begin{equation}\label{eq_defn_Gamma_a}
X_x \ni x \twoheadrightarrow  \Gamma_{\overline{a}}(x):= \left\{(a,y) \in \Gamma(x)~|~a \leq \overline{a}\right\}= \left\{(a,y)\in [\underline{a},\overline{a}] \times \K_y~|~\mathcal{I}_s(x,a,y) \geq 0 \text{ for all } s \in \mathcal{S}\right\}.
\end{equation}

\begin{lem}\label{lem_new}
Let $\overline{a}$ be defined in \eqref{eq_defn_overline_a}. Moreover, let $\K_x \ni x \mapsto \Gamma_{\overline{a}}(x)$ be defined in \eqref{eq_defn_Gamma_a}. Then, for all $x \in \K_x$, $\Gamma_{\overline{a}}(x)$ is nonempty, and for all $x \in \K_x$
\[
V_{\overline{a}}(x):= \inf_{(a,y) \in \Gamma_{\overline{a}}(x)} f(x,a,y)=\inf_{(a,y) \in \Gamma (x)} f(x,a,y)=V(x).
\]
 \end{lem}
\begin{proof}
By Remark~\ref{rem_assumptions}~(i) we see that $\Gamma_{\overline{a}}(x) \neq \emptyset$ for all $x \in \K_x$. Moreover, as $\Gamma_{\overline{a}}(x) \subseteq \Gamma(x)$, we have $V_{\overline{a}}(x)\geq V(x)$ for every $x \in \K_x$. To see that $V_{\overline{a}}(x)\leq V(x)$ for every $x \in \K_x$, fix any $x \in \K_x$ and let $(a,y) \in \Gamma(x)$. By Remark~\ref{rem_assumptions}~(i), we have $(\overline{a}^{\operatorname{UB}},y) \in \Gamma_{\overline{a}}(x)$. Hence, $f(x,a,y) \geq f(x,\min\{a,\overline{a}^{\operatorname{UB}}\},y) \geq \inf_{(\widetilde{a}, \widetilde{y}) \in \Gamma_{\overline{a}}(x)}f(x,\widetilde{a}, \widetilde{y})$. Since $(a,y) \in \Gamma(x)$ was arbitrary we obtain the desired result.
\end{proof}

\begin{lem} \label{lem_convex_2}
The map 
$
\K_x \ni x \twoheadrightarrow \Gamma_{\overline{a}}(x)$ defined in \eqref{eq_defn_Gamma_a}
is a non-empty, compact-valued, convex-valued, and continuous correspondence.
\end{lem}
\begin{proof}
The non-emptiness follows from Remark~\ref{rem_assumptions}.

Let $x \in \K_x$. Consider a sequence $(a^{(n)},y^{(n)})_{n \in \N} \subseteq \Gamma_{\overline{a}}(x)$. Then, by the compactness of $[\underline{a},\overline{a}]  \times \K_y$, there exists a subsequence $(a^{(n_k)},y^{(n_k)})_{k \in \N} \subseteq \Gamma_{\overline{a}}(x)$ such that $(a^{(n_k)},y^{(n_k)}) \rightarrow (a,y)$ as $k \rightarrow \infty$ for some $(a,y) \in [\underline{a},\overline{a}]  \times \K_y$. The continuity of  $[\underline{a},\overline{a}] \times \K_y \ni (a,y) \mapsto \mathcal{I}_s(x,a,y)$, which is ensured by Assumption~\ref{asu_IS}~(i), then implies  that 
$0 \leq \lim_{k \rightarrow \infty} \mathcal{I}_s(x,a^{(n_k)},y^{(n_k)}) = \mathcal{I}_s(x,a,y)$. Hence, $\Gamma_{\overline{a}}(x)$ is compact.

Let $x \in \K_x$, and let $(a,y),~(a',y') \in \Gamma_{\overline{a}}(x)$. Then, it follows for all $t \in [0,1]$ by Assumption~\ref{asu_IS}~(ii) that 
\[
\mathcal{I}_s \left(x,~t\cdot  a+(1-t)a',~ty+(1-t)\cdot y' \right)\geq t\cdot \mathcal{I}_s(x,a,y)+(1-t)  \cdot \mathcal{I}_s(x,a',y')\geq 0 \text{ for all } s \in \mathcal{S}.
\]
Hence, the convexity of \eqref{eq_defn_Gamma_a} follows.

It remains to show the continuity, i.e., that the map from \eqref{eq_defn_Gamma_a} is lower hemicontinuous and upper hemicontinuous.

Let $(x^{(n)})_{n \in \N}\subseteq \K_x$ with $\lim_{n \rightarrow \infty} x^{(n)} = x \in \K_x$ and let $(a,y) \in \Gamma_{\overline{a}}(x) \subseteq [\underline{a},\overline{a}]  \times \K_y$.
To show the lower-hemicontinuity, according to the characterization provided, e.g., in \cite[Theorem 17.21]{Aliprantis}, we need to prove the existence of a subsequence $(x^{(n_k)})_{k \in \N}$ and  elements $(a^{(k)},y^{(k)}) \in \Gamma(x^{(n_k)})$ for each $k \in \N$ with $\lim_{k \rightarrow \infty}(a^{(k)},y^{(k)})  =(a,y)$.

  First assume that $a \leq \overline{a}^{\operatorname{UB}}$. Since $\lim_{n \rightarrow \infty} x^{(n)} = x$, there exists, by definition of $\overline{a}$, some $n_0 \in \N$ such that for all $n \geq n_0$ we have 
\begin{equation}\label{eq_proof_lemma52_an}
a^{(n)}:= a+\frac{{L}_{\mathcal{I}}}{\mathcal{L}_{a,\mathcal{I}}} \cdot \left\|x^{(n)}-x\right\|\leq \overline{a}.
\end{equation}
Since by Assumption~\ref{asu_IS}~(iii) the map  $[\underline{a},\overline{a}]  \ni a \mapsto \mathcal{I}_s(x,a,y)$ is monotone for all $x \in \K_x, y \in \K_y, s \in \mathcal{S}$, with Assumption~\ref{asu_IS}~~(iv), and with the Lipschitz-property of $\mathcal{I}_s$ from Assumption~\ref{asu_IS}~(i),  we have  for all $s\in \mathcal{S}$ and for all $n \in \N$ that
\begin{align*}
\mathcal{I}_s\left(x^{(n)},~a^{(n)},~y\right) & = \mathcal{I}_s\left(x^{(n)},~a,~y\right) -\mathcal{I}_s\left(x^{(n)},~a,~y\right) +\mathcal{I}_s\left(x^{(n)},~a+\frac{{L}_{\mathcal{I}}}{\mathcal{L}_{a,\mathcal{I}}} \cdot \left\|x^{(n)}-x\right\|,~y\right) \\
&\geq \mathcal{I}_s\left(x^{(n)},~a,~y\right)+ \mathcal{L}_{a,\mathcal{I}} \cdot\frac{{L}_{\mathcal{I}}}{\mathcal{L}_{a,\mathcal{I}}}  \cdot \left\|x^{(n)}-x\right\| \\ 
&\geq \mathcal{I}_s\left(x^{(n)},~a,~y\right)-\mathcal{I}_s\left(x^{(n)},~a,~y\right)+\mathcal{I}_s\left(x,~a,~y\right) \geq 0,
\end{align*}
where the last inequality follows since $(a,y) \in \Gamma_{\overline{a}}(x)$. Thus, we have $(a^{(n)},y) \in \Gamma_{\overline{a}}(x^{(n)})$ for all $n \geq n_0$ as well as  by \eqref{eq_proof_lemma52_an} that $\lim_{n \rightarrow \infty} (a^{(n)},y) = (a,y)$. Hence lower-hemicontinuity follows for the case $a \leq \overline{a}^{\operatorname{UB}}$.\\
Now we consider the case that $a > \overline{a}^{\operatorname{UB}}$. Note that in this case $\mathcal{I}
_s(x,a,y) >0$ for all $s \in \mathcal{S}$ due to the strict monotonicity of $\mathcal{I}
_s$ and by Remark~\ref{rem_assumptions}~(i).
Hence, by the continuity of $\mathcal{I}
_s$, there exists some $n_0 \in \N$ such that for all $n \geq n_0$ we have 
$\mathcal{I}
_s(x^{(n)},a,y)>0$ implying that $(a,y) \in \Gamma_{\overline{a}}(x^{(n)})$ for all $n \geq n_0$.
Thus, we conclude with \cite[Theorem 17.21]{Aliprantis} the lower hemicontinuity of the map from \eqref{eq_defn_Gamma_a} also for the case  $a > \overline{a}^{\operatorname{UB}}$.

It remains to show the upper hemicontinuity. To this end, let $(x^{(n)},a^{(n)},y^{(n)}) \in \operatorname{Gr} \Gamma_{\overline{a}}$ with $\lim_{n \rightarrow} x^{(n)} = x$.
We apply  the characterization of upper hemicontinuity provided, e.g., in \cite[Theorem 17.20]{Aliprantis}, and therefore we need to show the existence of a subsequence $(a^{(n_k)},y^{(n_k)})_{k \in \N}$ with $\lim_{k \rightarrow \infty}( a^{(n_k)},y^{(n_k)})=(a,y) \in \Gamma_{\overline{a}}(x)$.

As $(a^{(n)},y^{(n)})_{n \in \N} \subseteq [\underline{a},\overline{a}]  \times \K_y$  is a sequence defined on a compact space, there exists a subsequence $(a^{(n_k)},y^{(n_k)})_{k \in \N}$ with   $\lim_{k \rightarrow \infty}( a^{(n_k)},y^{(n_k)})=(a,y) \in [\underline{a},\overline{a}]  \times \K_y$. Since $\mathcal{I}_s(x^{(n_k)},a^{(n_k)},y^{(n_k)}) \geq 0$ for all $k \in \N$ as $\left(x^{(n_k)},a^{(n_k)},y^{(n_k)}\right) \in \operatorname{Gr} \Gamma_{\overline{a}}$, we obtain by the continuity of $ \mathcal{I}_s$ that $\mathcal{I}_s(x,a,y) \geq 0$. This means $(a,y) \in \Gamma_{\overline{a}}(x)$.
\end{proof}

\begin{lem} \label{lem_convex_3}
For all $\varepsilon\in (0,1)$ the correspondence
\begin{equation}\label{eq_lem_convex_3}
\K_x \ni x \twoheadrightarrow \mathcal{M}_{\varepsilon}(x):=\left\{(a,y) \in \Gamma_{\overline{a}}(x)~\middle|~ f(x,a,y)- V_{\overline{a}}(x) < \varepsilon \right\}
\end{equation}
is non-empty, convex-valued, and lower hemicontinuous.
\end{lem}
\begin{proof}
Let $\varepsilon \in (0,1)$. The non-emptiness of $\mathcal{M}_{\varepsilon}(x)$ for each $x \in \K_x$  follows by definition and by Remark~\ref{rem_assumptions}.
To show the convexity of $\mathcal{M}_{\varepsilon}(x)$ for each $x \in \K_x$, fix any $x \in \K_x$ and let $(y,a),~ (\widetilde{y}, \widetilde{a}) \in \mathcal{M}_{\varepsilon}(x)$ and $t \in [0,1]$. Then by Lemma~\ref{lem_convex_2} implying that $\Gamma_{\overline{a}}(x)$ is convex, we have $t\cdot (a,y)+(1-t)\cdot (\widetilde{y}, \widetilde{a}) \in \Gamma_{\overline{a}}(x)$. Moreover, by Assumption~\ref{asu_f}~(iii) ensuring that $[\underline{a},\overline{a}] \times \K_y \ni (a,y) \mapsto f(x,a,y)$ is convex, we have
\begin{align*}
&f\left(x,t \cdot a + (1-t) \cdot \widetilde{a}, t \cdot y + (1-t) \cdot \widetilde{y}\right) - V_{\overline{a}}(x)\\ 
&\leq t \cdot \left( f\left(x, a,  y\right) - V_{\overline{a}}(x) \right)+ (1-t) \cdot \left(f\left(x, \widetilde{a},  \widetilde{y}\right) - V_{\overline{a}}(x)  \right) \leq t \cdot  \varepsilon +(1-t) \cdot \varepsilon = \varepsilon,
\end{align*}
from which we conclude the convexity of $\mathcal{M}_{\varepsilon}(x)$.
To show the lower hemicontinuity of \eqref{eq_lem_convex_3} let $(x^{(n)})_{n \in \N} \subseteq \K_x$ with $\lim_{n \rightarrow \infty}x^{(n)} =x  \in \K_x$, and let $(a,y) \in  \mathcal{M}_{\varepsilon}(x)$. We apply the characterization of lower hemicontinuity from \cite[Theorem 17.20]{Aliprantis} and therefore aim at showing that there exists a subsequence $(x^{(n_k)})_{k \in \N}$ and elements $(a^{(k)},y^{(k)}) \in \mathcal{M}_{\varepsilon}(x^{(n_k)})$ for each $k \in \N$ such that $\lim_{k \rightarrow \infty} (a^{(k)},y^{(k)}) = (a,y)$.

By Lemma~\ref{lem_convex_2} the correspondence $\K_x \ni x \twoheadrightarrow \Gamma_{\overline{a}}(x)$ is non-empty, compact-valued, continuous, and by Assumption~\ref{asu_f}~(ii), the map $ \K_x \times [\underline{a},\overline{a}] \times \K_y \ni(x,a,y) \mapsto f(x,a,y)$ is continuous. Hence, Berge's maximum theorem (see \cite{berge} or \cite[Theorem 17.31]{Aliprantis}) is applicable. 

We then obtain by Berge's maximum theorem that the map 
\[
\K_x \ni x \mapsto V_{\overline{a}}(x):=\inf_{(a,y) \in \Gamma_{\overline{a}}(x)} f(x,a,y)
\]
is continuous. Therefore, as $(a,y) \in \mathcal{M}_{\varepsilon}(x)$, and since both $f$ and $V_{\overline{a}}$ are continuous, there exists some $\gamma \in (0,1)$ such that for all $(x,'a',y')$ with $(x,'a',y') \in \mathcal{B}_{\gamma}(x,a,y) \subseteq \K_x \times [\underline{a},\overline{a}] \times \K_y$, it holds 
\begin{equation}\label{eq_f_psi_eps}
f(x',a',y') - V_{\overline{a}}(x') < \varepsilon.
\end{equation}
Moreover, as $\lim_{n \rightarrow \infty} x^{(n)} = x$, there exist some $n_0 \in \N$ such that for all $n \geq n_0$ we have 
\begin{equation}\label{eq_sqrt_gamma}
\sqrt{\left\|x^{(n)}-x\right\|^2+\left(\frac{L_{\mathcal{I}}}{\mathcal{L}_{a,\mathcal{I}}}\left\|x^{(n)}-x\right\|\right)^2} \leq \gamma.
\end{equation}
Moreover, since $(a,y) \in \mathcal{M}_{\varepsilon}(x)$ and $\varepsilon \in (0,1)$, we have by Lemma~\ref{lem_convex_1}~(ii) that $a < \overline{a}^{\operatorname{UB}}+\tfrac{1}{\mathcal{L}_{a,f}}$.
Hence, by \eqref{eq_sqrt_gamma} and by definition of $\overline{a}$ we have for all $n \geq n_0$ also that
\begin{equation}\label{eq_proof_lem_53_eq_1}
a+ \frac{L_{\mathcal{I}}}{\mathcal{L}_{a,\mathcal{I}}} \left\| x^{(n)}-x\right\| \leq a+ \gamma \leq \overline{a}.
\end{equation}
Note also that for  $n \geq n_0$ we have by Assumption~\ref{asu_IS}~(iv) and Assumption~\ref{asu_IS}~(i)   for all $ s \in \mathcal{S}$ the following inequality
\begin{equation} \label{eq_proof_lem_53_eq_2}
\begin{aligned}
\mathcal{I}_s \left(x^{(n)},~a+ \frac{L_{\mathcal{I}}}{\mathcal{L}_{a,\mathcal{I}}} \left\| x^{(n)}-x\right\|,~y\right) & = \mathcal{I}_s \left(x^{(n)},~a+ \frac{L_{\mathcal{I}}}{\mathcal{L}_{a,\mathcal{I}}} \left\| x^{(n)}-x\right\|,~y\right)-\mathcal{I}_s \left(x^{(n)},~a,~y\right)\\
&\hspace{6.2cm}+\mathcal{I}_s \left(x^{(n)},~a,~y\right)\\
&\geq \mathcal{L}_{a,\mathcal{I}} \frac{L_{\mathcal{I}}}{\mathcal{L}_{a,\mathcal{I}}} \left\|x^{(n)}-x\right\|+\mathcal{I}_s \left(x^{(n)},~a,~y\right)\\
&= {L_{\mathcal{I}}} \left\|x^{(n)}-x\right\|+\mathcal{I}_s \left(x^{(n)},~a,~y\right)-\mathcal{I}_s \left(x,~a,~y\right)+\mathcal{I}_s \left(x,~a,~y\right)\\
&\geq L_{\mathcal{I}} \left\|x^{(n)}-x\right\|-L_{\mathcal{I}} \left\|x^{(n)}-x\right\| +\mathcal{I}_s \left(x,~a,~y\right) \geq 0,
\end{aligned}
\end{equation}
since $(a,y) \in \Gamma_{\overline{a}}(x)$. Hence, \eqref{eq_proof_lem_53_eq_1} and \eqref{eq_proof_lem_53_eq_2} together show that
\begin{equation}\label{eq_stratingamma}
\left(a+ \frac{L_{\mathcal{I}}}{\mathcal{L}_{a,\mathcal{I}}} \left\| x^{(n)}-x\right\|,~y\right) \in \Gamma_{\overline{a}}(x^{(n)}) \text{ for all } n \geq n_0.
\end{equation}
By \eqref{eq_sqrt_gamma} we have $\left(x^{(n)},a+ \frac{L_{\mathcal{I}}}{\mathcal{L}_{a,\mathcal{I}}} \left\| x^{(n)}-x\right\|,~y\right) \in \mathcal{B}_{\gamma}(x,a,y)$ for all $n \geq n_0$. Thus, it follows with \eqref{eq_f_psi_eps} and \eqref{eq_stratingamma} that
\[
\left(a+ \frac{L_{\mathcal{I}}}{\mathcal{L}_{a,\mathcal{I}}} \left\| x^{(n)}-x\right\|,~y\right) \in \mathcal{M}_{\varepsilon}\left(x^{(n)}\right)\text{ for all } n \geq n_0,
\]
proving the lower hemicontinuity of \eqref{eq_lem_convex_3}, by applying the characterization of lower hemicontinuity from \cite[Theorem 17.20]{Aliprantis} to the subsequences $(x^{(n)})_{n\in \N,\atop n\geq n_0}$ and $\left(a+ \frac{L_{\mathcal{I}}}{\mathcal{L}_{a,\mathcal{I}}} \left\| x^{(n)}-x\right\|,~y\right)_{n\in \N,\atop n\geq n_0}$.
\end{proof}
\begin{cor}\label{cor_convex_1}
For all $\varepsilon\in (0,1)$ the correspondence
\begin{equation}\label{eq_cor_convex_1_1}
\K_x \ni x \twoheadrightarrow \overline{\mathcal{M}_{\varepsilon}(x)} := \operatorname{cl}\left(\mathcal{M}_{\varepsilon}(x)\right)
\end{equation}
is nonempty, convex, closed, lower hemicontinuous, and satisfies
\begin{equation}\label{eq_cor_convex_1_2}
\overline{\mathcal{M}_{\varepsilon}(x)} \subseteq \left\{(a,y) \in \Gamma_{\overline{a}}(x)~\middle|~ f(x,a,y)-V_{\overline{a}}(x) \leq \varepsilon\right\}.
\end{equation}
\end{cor}
\begin{proof}
The non-emptiness and convexity  of the map defined in \eqref{eq_cor_convex_1_1} both follow from Lemma~\ref{lem_convex_3}. That the map is closed is a consequence of the definition of a closure of a set. The lower-hemicontinuity also follows from Lemma~\ref{lem_convex_3} and from \cite[Theorem 17.22~(1), p. 566]{Aliprantis} which ensures that the closure of a lower hemicontinuous map is again lower hemicontinuous. The relation \eqref{eq_cor_convex_1_2} follows as the map $ \K_x \times [\underline{a},\overline{a}] \times \K_y \ni(x,a,y) \mapsto f(x,a,y)$ is continuous by Assumption~\ref{asu_f}~(ii).
\end{proof}

\begin{cor}\label{cor_convex_2}
For all $\varepsilon \in (0,1)$ there exists a continuous map $\K_x \ni x \mapsto  \left(a^{*,\varepsilon}(x),~y^{*,\varepsilon}(x)\right) \in \Gamma_{\overline{a}}(x)$ satisfying both
\begin{itemize}
\item[(i)] $a^{*,\varepsilon}(x) \leq \overline{a}^{\operatorname{UB}}+\frac{1}{\mathcal{L}_{a,f}}$ for all $x \in \K_x$,
\item[(ii)]
$f\left(x,~a^{*,\varepsilon}(x),~y^{*,\varepsilon}\right)-V_{\overline{a}}(x) \leq \varepsilon$.
\end{itemize}
\end{cor}

\begin{proof} Corollary~\ref{cor_convex_1} ensures that the requirements for an application of the Michael selection theorem (see \cite{michael} or \cite[Theorem 17.66]{Aliprantis}) are fulfilled. By the Michael selection theorem we then obtain a continuous selector $\K_x \ni x \mapsto  \left(a^{*,\varepsilon}(x),~y^{*,\varepsilon}(x)\right) \in \operatorname{cl}\left(\mathcal{M}_{\varepsilon}(x)\right) \subseteq \Gamma_{\overline{a}}(x)$  implying, by definition of $\mathcal{M}_{\varepsilon}(x)$, that (ii) is fulfilled.

Assume now that (i) does not hold, i.e., that we have $a^{*,\varepsilon}(x) > \overline{a}^{\operatorname{UB}}+\frac{1}{\mathcal{L}_{a,f}}$. This, however, by Lemma~\ref{lem_convex_1}~(ii), contradicts (ii), which concludes the proof.
\end{proof}
%
%
%

Now, for any $0 < \delta < r$ recall the definition of the set $C_{y,\delta} \subseteq \K_y $ from Assumption~\ref{asu_KxKy}. 
\begin{lem} \label{lem_convex_4}
For all $\delta\in (0,r)$, the map 
\[
\K_x \ni x \mapsto \big({a_\delta}^{*,\varepsilon}(x),~ {y_\delta}^{*,\varepsilon}(x)\big):= \underset{(a,y) \in [\underline{a}+\delta,\overline{a}-\delta] \times C_{y,\delta}}{\operatorname{argmin}} \left\|(a,y)-\left(a^{*,\varepsilon}(x),y^{*,\varepsilon}(x)\right)\right\|^2
\]
is continuous.
\end{lem}

\begin{proof} 
Note that $(x,a,y) \mapsto\left\|(a,y)-\left(a^{*,\varepsilon}(x),y^{*,\varepsilon}(x)\right)\right\|^2$ is continuous by Corollary~\ref{cor_convex_2}.
Moreover, the single-valued map is well-defined as the projection of the point $\big({a}^{*,\varepsilon}(x),~ y^{*,\varepsilon}(x)\big)$ onto the compact, convex set $[\underline{a}+\delta,\overline{a}-\delta] \times  C_{y,\delta} $. The continuity follows now by, e.g., Berge's maximum theorem (\cite[Theorem 17.31]{Aliprantis}) and \cite[Lemma 17.6]{Aliprantis}.
\end{proof}
\subsubsection{Proof of Theorem~\ref{thm_main}}

In Section~\ref{sec_auxiliary} we have established all auxiliary results that allow us now to report the proof of Theorem~\ref{thm_main}.
\begin{proof}[Proof of Theorem~\ref{thm_main}]
Without loss of generality let $\varepsilon \in (0,1)$, else we substitute $\varepsilon$ by $\overline{\varepsilon}:=\frac{\varepsilon}{1+\varepsilon}$. By Corollary~\ref{cor_convex_2}, for all $x \in \K_x$ there exists, by abuse of notation with $\varepsilon \leftarrow \varepsilon/2$ in the notation of Corollary~\ref{cor_convex_2}, some continuous map $\K_x \ni x \mapsto \left(a^{*,\varepsilon}(x),~y^{*,\varepsilon}(x)\right) \in \Gamma_{\overline{a}}(x)$ satisfying for all $x \in \K_x$ that 
\begin{equation}\label{eq_condition_astar}
a^{*,\varepsilon}(x) \leq \overline{a}^{\operatorname{UB}}+\frac{1}{\mathcal{L}_{a,f}}
\end{equation}
and  such that 
\begin{equation}
\label{eq_condition_astar_2}
f\left(x,~a^{*,\varepsilon}(x),~y^{*,\varepsilon}\right)-V_{\overline{a}}(x) \leq \varepsilon/2.
\end{equation}
 We recall $r \in (0,1)$ from Assumption~\ref{asu_KxKy} and define 
 \begin{equation}\label{eq_def_delta_0}
 \delta_0 := \frac{\varepsilon \min \{\mathcal{L}_{a,\mathcal{I}},1\}}{8 \max \left\{L_{\mathcal{I}},\mathcal{L}_{a,f}\right\} \sqrt{1+L_r^2} L_f} \cdot r \in (0,r). 
 \end{equation}
Note that by definition of the projection from Lemma~\ref{lem_convex_4} with respect to $[\underline{a}-\delta_0, \underline{a}+\delta_0] \times C_{y,\delta_0}$, and by Assumption~\ref{asu_KxKy}  we have
\begin{equation}\label{eq_diff_arguments}
\left \| \left(x,a^{*,\varepsilon}(x),~y^{*,\varepsilon}(x)\right)- \left(x,a_{\delta_0}^{*,\varepsilon}(x),~y_{\delta_0}^{*,\varepsilon}(x)\right)\right\| \leq \sqrt{\delta_0^2+L_r^2 \delta_0^2} = \delta_0 \sqrt{1+L_r^2} \text{ for all } x \in \K_x.
\end{equation}
Hence, for all $x \in \K_x$, by using the Lipschitz-continuity of $f$ from Assumption~\ref{asu_f}~(i), by \eqref{eq_diff_arguments}, and by the definition of $\delta_0$ in \eqref{eq_def_delta_0}, we have
\begin{equation}\label{eq_diff_f}
\begin{aligned}
\left|f \left(x,a^{*,\varepsilon}(x),~y^{*,\varepsilon}(x)\right)-f \left(x,a_{\delta_0}^{*,\varepsilon}(x),~y_{\delta_0}^{*,\varepsilon}(x)\right)\right| &\leq L_f \left \| \left(x,a^{*,\varepsilon}(x),~y^{*,\varepsilon}(x)\right)- \left(x,a_{\delta_0}^{*,\varepsilon}(x),~y_{\delta_0}^{*,\varepsilon}(x)\right)\right\| \\
&\leq   L_f \cdot \delta_0 \sqrt{1+L_r^2} \\
& = r\cdot \frac{L_f}{L_f} \frac{ \sqrt{1+L_r^2}}{ \sqrt{1+L_r^2}}\cdot \frac{\varepsilon \min \{\mathcal{L}_{a,\mathcal{I}},1\}}{8 \max \left\{L_{\mathcal{I}},\mathcal{L}_{a,f}\right\}}\leq \frac{\varepsilon}{8}.
\end{aligned}
\end{equation}
By Corollary~\ref{cor_convex_1}, we have $(x,a^{*,\varepsilon}(x),y^{*,\varepsilon}(x)) \in \Gamma_{\overline{a}}(x)$, and in particular, $\mathcal{I}_s(x,a^{*,\varepsilon}(x),y^{*,\varepsilon}(x)) \geq 0$ for all $s \in  \mathcal{S}$ and all $x\in \K_x$. This implies by the Lipschitz-continuity of $\mathcal{I}_s$ (Assumption~\ref{asu_IS}~(i)), by using \eqref{eq_diff_arguments}, and the definition of $\delta_0$,  that
\begin{equation}\label{eq_ineq_Is_proof_1} 
\begin{aligned}
\mathcal{I}_s(x,~a_{\delta_0}^{*,\varepsilon}(x),~y_{\delta_0}^{*,\varepsilon}(x))  &=\mathcal{I}_s(x,~a_{\delta_0}^{*,\varepsilon}(x),~y_{\delta_0}^{*,\varepsilon}(x))-\mathcal{I}_s(x,~a^{*,\varepsilon}(x),~y^{*,\varepsilon}(x))+ \mathcal{I}_s(x,~a^{*,\varepsilon}(x),~y^{*,\varepsilon}(x))\\
&\geq - L_{\mathcal{I}}\left\| \left(x,a^{*,\varepsilon}(x),~y^{*,\varepsilon}(x)\right)- \left(x,a_{\delta_0}^{*,\varepsilon}(x),~y_{\delta_0}^{*,\varepsilon}(x)\right)\right\|\\
&\geq - L_{\mathcal{I}} \delta_0 \sqrt{1+L_r^2} \\
&=  -  r \cdot \frac{\sqrt{1+L_r^2}} { \sqrt{1+L_r^2}}  \cdot\frac{\min \{\mathcal{L}_{a,\mathcal{I}},1\}}{L_f}  \cdot \frac{L_{\mathcal{I}}}{\max \left\{L_{\mathcal{I}},\mathcal{L}_{a,f}\right\}}  \cdot \frac{\varepsilon }{8  }  \geq -\frac{\mathcal{L}_{a,\mathcal{I}}}{L_f} \frac{\varepsilon}{8}.
\end{aligned}
\end{equation}
By the universal approximation theorem (Proposition~\ref{lem_universal}) and Lemma~\ref{lem_convex_4} there exists a neural network 
$
 \widetilde{\NN}:= \left(\widetilde{\NN}_a,\widetilde{\NN}_y\right)\in  \mathfrak{N}_{n_x,1+n_y}
$
such that
\begin{equation}\label{eq_sup_nn_small}
\sup_{x \in \K_x} \left\| \left({a}_{\delta_0}^{*,\varepsilon}(x),~{y}_{\delta_0}^{*,\varepsilon}(x)\right)-\left(\widetilde{\NN}_a(x),\widetilde{\NN}_y(x)\right)  \right\| < \delta_0.
\end{equation}
Moreover, we have by \eqref{eq_ineq_Is_proof_1} and \eqref{eq_sup_nn_small} for all $x\in \K_x$, $s \in \mathcal{S}$ that
\begin{equation}\label{eq_ariel_6}
\begin{aligned}
\mathcal{I}_s(x,\widetilde{\NN}_a(x),\widetilde{\NN}_y(x))  &=\mathcal{I}_s(x,\widetilde{\NN}_a(x),\widetilde{\NN}_y(x))-\mathcal{I}_s(x,{a}_{\delta_0}^{*,\varepsilon}(x),~{y}_{\delta_0}^{*,\varepsilon}(x))+\mathcal{I}_s(x,{a}_{\delta_0}^{*,\varepsilon}(x),~{y}_{\delta_0}^{*,\varepsilon}(x))\\
&\geq -L_{\mathcal{I}} \delta_0 +\mathcal{I}_s(x,{a}_{\delta_0}^{*,\varepsilon}(x),~{y}_{\delta_0}^{*,\varepsilon}(x)) \\
&\geq -L_{\mathcal{I}}\frac{\varepsilon \min \{\mathcal{L}_{a,\mathcal{I}},1\}}{8 \max \left\{L_{\mathcal{I}},\mathcal{L}_{a,f}\right\} \sqrt{1+L_r^2} L_f} \cdot r -\frac{\mathcal{L}_{a,\mathcal{I}}}{L_f} \frac{\varepsilon}{8}  \\
&\geq -\frac{\mathcal{L}_{a,\mathcal{I}}}{L_f} \frac{\varepsilon}{8}  -\frac{\mathcal{L}_{a,\mathcal{I}}}{L_f} \frac{\varepsilon}{8} = -\frac{\mathcal{L}_{a,\mathcal{I}}}{L_f} \frac{\varepsilon}{4}.
\end{aligned}
\end{equation}
In addition, we have by \eqref{eq_sup_nn_small} and Assumption~\ref{asu_KxKy} that $\left(\widetilde{\NN}_a(x),\widetilde{\NN}_y(x)\right) \in [\underline{a},\overline{a}]\times \K_y $ for all $x \in \K_x$. Furthermore, for all $x \in \K_x$
\begin{equation}
\label{eq_ariel_8}
\begin{aligned}
\left|f\left(x, \widetilde{\NN}_a(x),\widetilde{\NN}_y(x) \right)-f \left(x, {a}_{\delta_0}^{*,\varepsilon}(x),~{y}_{\delta_0}^{*,\varepsilon}(x)\right)\right| \leq L_f \delta_0  = r \cdot \frac{L_f}{L_f} \frac{\varepsilon \min \{\mathcal{L}_{a,\mathcal{I}},1\}}{8 \max \left\{L_{\mathcal{I}},\mathcal{L}_{a,f}\right\} \sqrt{1+L_r^2}} \leq \frac{\varepsilon}{8}.
\end{aligned}
\end{equation}
Next, define  a neural network $\NN:= \left({\NN}_a,{\NN}_y\right) \in  \mathfrak{N}_{n_x,1+n_y}$ by
\begin{equation}\label{eq_ariel_9}
\left({\NN}_a(x),{\NN}_y(x)\right):=\left(\widetilde{\NN}_a(x)+\frac{1}{L_f} \frac{\varepsilon}{4},\widetilde{\NN}_y(x)\right),\qquad x \in \R^{n_x}.
\end{equation}
Then, for all $x\in \K_x$, by using \eqref{eq_sup_nn_small}, \eqref{eq_diff_arguments}, \eqref{eq_def_delta_0}, Corollary~\ref{cor_convex_2}, and the definition of $\overline{a}$ in \eqref{eq_defn_overline_a}, we obtain
\begin{align}\label{eq_ariel_10}
\NN_a(x) &= \widetilde{\NN}_a(x)+\frac{1}{L_f} \frac{\varepsilon}{4} \leq {a}_{\delta_0}^{*,\varepsilon}(x)+ \delta_0 +\frac{1}{L_f} \frac{\varepsilon}{4} \leq {a}^{*,\varepsilon}(x)+ \delta_0 \sqrt{1+L_r^2}+ \delta_0 +\frac{1}{L_f} \frac{\varepsilon}{4} \\
&\leq  \overline{a}^{\operatorname{UB}}+\frac{1}{\mathcal{L}_{a,f}}+ \delta_0 \sqrt{1+L_r^2}+ \delta_0 +\frac{1}{L_f} \frac{\varepsilon}{4} \\
&= \overline{a}^{\operatorname{UB}}+\frac{1}{\mathcal{L}_{a,f}}+ \frac{\varepsilon \min \{\mathcal{L}_{a,\mathcal{I}},1\}}{8 \max \left\{L_{\mathcal{I}},\mathcal{L}_{a,f}\right\} \sqrt{1+L_r^2} L_f} \cdot r  \sqrt{1+L_r^2}+ \delta_0 +\frac{1}{L_f} \frac{\varepsilon}{4}\\
&\leq \overline{a}^{\operatorname{UB}}+\frac{1}{\mathcal{L}_{a,f}}+ \frac{\varepsilon}{8} + \delta_0 +\frac{\varepsilon}{4} \leq \overline{a}^{\operatorname{UB}}+\frac{1}{\mathcal{L}_{a,f}}+2 \leq \overline{a}.
\end{align}
Hence, we conclude by \eqref{eq_ariel_9} and \eqref{eq_ariel_10} that 
\begin{equation}\label{eq_ariel_11}
\left({\NN}_a(x),{\NN}_y(x)\right) \in [\underline{a},\overline{a}]\times \K_y \text{ for all } x \in \K_x.
\end{equation}
Moreover, by \eqref{eq_ariel_9} and \eqref{eq_ariel_6} we have for all $x \in \K_x$ and $s \in \mathcal{S}$ that
\begin{equation}\label{eq_ariel_12}
\begin{aligned}
\mathcal{I}_s(x,{\NN}_a(x),{\NN}_y(x)) &=\mathcal{I}_s(x,{\NN}_a(x),\widetilde{\NN}_y(x))-\mathcal{I}_s(x,\widetilde{\NN}_a(x),\widetilde{\NN}_y(x))\\
&\hspace{4.5cm}+\mathcal{I}_s(x,\widetilde{\NN}_a(x),\widetilde{\NN}_y(x))\\
&\geq \frac{\mathcal{L}_{a,\mathcal{I}}}{L_f}\frac{\varepsilon}{4}+\mathcal{I}_s(x,\widetilde{\NN}_a(x),\widetilde{\NN}_y(x))\\
&\geq \frac{\mathcal{L}_{a,\mathcal{I}}}{L_f}\frac{\varepsilon}{4}-\frac{\mathcal{L}_{a,\mathcal{I}}}{L_f} \frac{\varepsilon}{4}=0.
\end{aligned}
\end{equation}
Hence, we see that 
\begin{equation}\label{eq_ariel_13}
\begin{aligned}
\left({\NN}_a(x),{\NN}_y(x)\right) \in \Gamma_{\overline{a}}(x) \subseteq \Gamma(x) \text{ for all } x \in \K_x.
\end{aligned}
\end{equation}
Furthermore, by \eqref{eq_ariel_9}, we have for all $x \in \K_x$ that
\begin{equation}\label{eq_ariel_14}
\begin{aligned}
\left|f\left(x, \widetilde{\NN}_a(x),\widetilde{\NN}_y(x) \right)-f\left(x, \NN_a(x),\NN_y(x) \right)\right| \leq L_f \frac{1}{L_f} \frac{\varepsilon}{4}= \frac{\varepsilon}{4}.
\end{aligned}
\end{equation}
Therefore, we conclude by Lemma~\ref{lem_new}, \eqref{eq_condition_astar_2}, \eqref{eq_diff_f}, \eqref{eq_ariel_8}, and \eqref{eq_ariel_14} that for all $x \in \K_x$ 
\begin{align*}
f\left(x, \NN_a(x),\NN_y(x) \right) -V(x) =& f\left(x, \NN_a(x),\NN_y(x) \right) -V_{\overline{a}}(x)\\
=&\left(f\left(x, \NN_a(x),\NN_y(x) \right)-f\left(x, \widetilde{\NN}_a(x),\widetilde{\NN}_y(x) \right) \right)\\
&+\left(f\left(x, \widetilde{\NN}_a(x),\widetilde{\NN}_y(x) \right) -f \left(x, {a}_{\delta_0}^{*,\varepsilon}(x),~{y}_{\delta_0}^{*,\varepsilon}(x)\right)\right) \\&+ \left(f \left(x, {a}_{\delta_0}^{*,\varepsilon}(x),~{y}_{\delta_0}^{*,\varepsilon}(x)\right)-f \left(x,a^{*,\varepsilon}(x),~y^{*,\varepsilon}(x)\right)\right)\\
&+\left(f \left(x,a^{*,\varepsilon}(x),~y^{*,\varepsilon}(x)\right)-V_{\overline{a}}(x)\right) \\
\leq &\frac{\varepsilon}{4}+\frac{\varepsilon}{8}+\frac{\varepsilon}{8}+\frac{\varepsilon}{2} = \varepsilon.
\end{align*}
\end{proof}
%
%
\vspace{0.6cm}
\section*{Acknowledgments}
Financial support by the Nanyang Assistant Professorship Grant (NAP Grant) \textit{Machine Learning based Algorithms in Finance and Insurance} is
gratefully acknowledged.
\vspace{0.6cm}
\bibliographystyle{ecta} 
\bibliography{literature}
\end{document}